\documentclass[12pt]{article}
\usepackage{amsmath,amssymb,amsfonts,amsthm,bbm}
\usepackage{graphicx}
\usepackage{enumerate}
\usepackage{natbib}
\usepackage{url}
\usepackage[font=small]{caption}
\usepackage{subcaption}
\usepackage{float}
\usepackage[T1]{fontenc}

\RequirePackage[colorlinks,citecolor=blue,urlcolor=red]{hyperref}

\makeatletter
\def\singlespace{\def\baselinestretch{1}\@normalsize}

\makeatletter
\def\singlespace{\def\baselinestretch{1}\@normalsize}

\numberwithin{equation}{section}

\renewcommand{\hat}{\widehat}

\renewcommand{\hat}{\widehat}

\newcommand{\bfm}[1]{\ensuremath{\mathbf{#1}}}

\def\ba{\bfm a}   \def\bA{\bfm A}  
\def\bb{\bfm b}   \def\bB{\bfm B}  
   \def\bC{\bfm C}  
\def\bd{\bfm d}   \def\bD{\bfm D}  
     
\def\bff{\bfm f}  \def\bF{\bfm F}

   \def\bI{\bfm I}

   \def\bM{\bfm M}

     \def\PP{\mathbb{P}}
     
   \def\bR{\bfm R}

\def\bu{\bfm u}   \def\bU{\bfm U}  
\def\bv{\bfm v}   \def\bV{\bfm V}  
\def\bw{\bfm w}   \def\bW{\bfm W}  
\def\bx{\bfm x}   \def\bX{\bfm X}  
     
   \def\bZ{\bfm Z}

\def\calF{{\cal  F}}

 \def\cL{{\cal  L}}

 \def\cP{{\cal  P}}

 \def\cS{{\cal  S}}

\newcommand{\bfsym}[1]{\ensuremath{\boldsymbol{#1}}}

 \def\bbeta{\bfsym \beta}
 \def\bgamma{\bfsym \gamma}             \def\bGamma{\bfsym \Gamma}
 \def\bdelta{\bfsym {\delta}}           
               
 \def\bmu{\bfsym {\mu}}                 
 	
 \def\btheta{\bfsym {\theta}}           
           
              \def\bSigma{\bfsym \Sigma}
              
         \def\bLambda {\bfsym {\Lambda}}

 \def\bxi{\bfsym {\xi}}
 
\def \R {\mathbb{R}}

 \def\hbeta{\hat{\beta}}

\DeclareMathOperator*{\argmin}{argmin}

\DeclareMathOperator{\cov}{cov}

\DeclareMathOperator{\diag}{diag}

\DeclareMathOperator{\E}{\mathbb{E}}

\DeclareMathOperator{\supp}{supp}

\DeclareMathOperator{\sign}{\rm sign}

\def\newpage{\vfill\eject}

\def\today{\ifcase\month\or
  January\or February\or March\or April\or May\or June\or
  July\or August\or September\or October\or November\or December\fi
  \space\number\day, \number\year}

\newdimen\biblioindent    \biblioindent=30pt

 at 8truept

\newcommand{\beq}{\begin{equation}}
  \newcommand{\eeq}{\end{equation}}
\newcommand{\beqn}{\begin{eqnarray}}
  \newcommand{\eeqn}{\end{eqnarray}}
\newcommand{\beqnn}{\begin{eqnarray*}}
  \newcommand{\eeqnn}{\end{eqnarray*}}

\allowdisplaybreaks
\usepackage{verbatim}
\renewcommand{\baselinestretch}{1.4}
\numberwithin{equation}{section}
\theoremstyle{plain}
\newtheorem{thm}{Theorem}[section]

\newtheorem{lem}{Lemma}[section]

\newtheorem{ass}{Assumption}[section]
\theoremstyle{definition}

\newtheorem{rem}{Remark}[section]

\newcounter{CondCounter}

\begin{document}

\def\spacingset#1{\renewcommand{\baselinestretch}%
{#1}\small\normalsize} \spacingset{1}


  \title{\bf Factor-Augmented Regularized Model for Hazard Regression}
  \author{Pierre Bayle\thanks{
  The authors gratefully acknowledge the support of NIH grant 2R01-GM072611-14 and NSF grant DMS-2053832. Emails: \texttt{pbayle@princeton.edu}, \texttt{jqfan@princeton.edu}.
  }\hspace{.2cm} and Jianqing Fan\footnotemark[1]\hspace{.2cm}\\
    Department of Operations Research and Financial Engineering\\
     Princeton University}
  \date{September 30, 2022}
  \maketitle
  \vspace*{-15mm}

\bigskip
\begin{abstract}
A prevalent feature of high-dimensional data is the dependence among covariates, and model selection is known to be challenging when covariates are highly correlated.
To perform model selection for the high-dimensional Cox proportional hazards model in presence of  correlated covariates with factor structure, we propose a new model, Factor-Augmented Regularized Model for Hazard Regression (FarmHazard), which builds upon latent factors that drive covariate dependence and extends Cox's model. This new model generates procedures that operate in two steps by learning factors and idiosyncratic components from high-dimensional covariate vectors and then using them as new predictors.
Cox's model is a widely used semi-parametric model for survival analysis, where censored data and time-dependent covariates bring additional technical challenges.
We prove model selection consistency and estimation consistency under mild conditions.
We also develop a factor-augmented variable screening procedure to deal with strong correlations in ultra-high dimensional problems.
Extensive simulations and real data experiments demonstrate that our procedures enjoy good performance and achieve better results on model selection, out-of-sample C-index and screening than alternative methods.
\end{abstract}

\noindent%
{\it Keywords:}  High-dimensional,  Cox's proportional hazards model, Factor model, Model selection, Censored data, Variable screening.
\vfill

\newpage
\spacingset{1.5}
\section{Introduction}
\label{sec:intro}

An enormous volume of data is accessible in many fields, including biomedicine and clinical trials, and efficient and valid statistical methods are necessary to study it.
In survival analysis, the outcome variable is time-to-event, such as biological death, relapse, failure of a mechanical engine or credit default, and often some observations are censored. For example, a study can come to an end while a fraction of subjects have not experienced the event of interest, or a subject can leave the study before its end. In this context, a widely used semi-parametric model is Cox's proportional hazards model \citep{Cox_1972, Cox_1975}. \citet{Andersen_Gill_1982} formulated it into a counting process framework. In the fixed-dimension setting, \cite{Tsiatis_1981} and \citet{Andersen_Gill_1982} proved the consistency and asymptotic normality of the maximum partial likelihood estimator. Yet, modern datasets frequently have more predictors than samples.
Despite the very large number of predictors, most of them are often irrelevant to explain the outcome, leading to sparse models. Reducing high-dimensional data to the true set of relevant covariates is one of the most important tasks in high-dimensional statistics and is a challenge in the analysis of big data. Models would become more interpretable and prediction more accurate. To this end, several regularized regression techniques have been extended to Cox's proportional hazards model \citep{Tibshirani_1997, Fan_Li_2002}. \citet{Bradic_etal_2011} established model selection consistency and strong oracle properties for a large class of penalty functions in the ultra-high dimensional setting, with LASSO and SCAD as special cases. \citet{Huang_etal_2013} and \citet{Kong_Nan_2014} studied oracle inequalities for LASSO under different conditions.

When variables are correlated, most model selection techniques fail to recover the set of important predictors, in both high-dimensional and ultra-high dimensional settings. \citet{Fan_etal_2020} suggested FarmSelect, a two-step procedure that learns factors and idiosyncratic components and use them as new predictors, to overcome the dependence problem among covariates in the setting of $\ell_{1}$-penalized generalized linear models.
An even more demanding task is to consider models that go beyond generalized linear models, such as Cox's proportional hazards model, where censored data and time-dependent covariates bring additional technical challenges.
To cope with correlation in high dimensions within the challenging survival analysis setting, we propose Factor-Augmented Regularized Model for Hazard Regression (FarmHazard). High-dimensional genomics and genetic data are naturally strongly correlated, and our model is designed to address this kind of issues. It has important applications, such as the prediction of the outcome of chemotherapy based on gene-expression profiles coming from DNA microarrays \citep{Rosenwald_etal_2002}.

In ultra-high dimensional problems, characterized by a dimension that grows with the sample size in a non-polynomial fashion, regularized regression faces multiple statistical and computational challenges \citep{Fan_etal_2009}.
To remedy this, screening methods \citep{Fan_Lv_2008, Fan_Song_2009, Wang_Leng_2016} have been developed; they enjoy statistical guarantees and are computationally efficient. \citet{Fan_etal_2010} extended the key idea of sure independence screening to Cox's model, and \citet{Zhao_Li_2012} provided theoretical support. Yet, screening methods tend to include too many variables when strong correlations exist among covariates \citep{Fan_Lv_2008, Wang_Leng_2016}. We propose a factor-augmented variable screening procedure that is able to deal with these strong correlations for Cox's proportional hazards model.

The paper is organized as follows. Section~\ref{sec:setup} formulates the problem. In Section~\ref{sec:pre-theory}, we introduce FarmHazard and present properties of the estimated factors and idiosyncratic components. We provide the main theoretical guarantees in Section~\ref{sec:theory}, and perform extensive simulations and real data experiments in Section~\ref{sec:num-exp}. Proofs of the various results can be found in the Appendix.

We introduce a few notations used throughout the paper. For any integer $n$, we denote $[n] = \{1, \dots, n\}$. $\bI_n$ denotes the $n\times n$ identity matrix and $\mathbf{0}_n$ represents the all-zero vector in $\R^n$. For a vector $\bgamma = (\gamma_{1}, \dots, \gamma_{m})^{\top} \in \R^m$ and $q \in \mathbb{N}^{\star}$, denote the $\ell_{q}$ norm $\|\bgamma\|_{q} = (\sum_{i = 1}^{m} |\gamma_{i}|^{q})^{1/q}$ and $\|\bgamma\|_{\infty} = \max\limits_{i \leq m} |\gamma_{i}|$. The support set $\supp(\bgamma)$ is $\{i \in [m]: \gamma_i \neq 0\}$, and $\sign(\bgamma)$ is the vector $(\sign(\gamma_i))_{i \in [m]}$, where $\sign(\gamma_i)=1$, $0$, or $-1$ for $\gamma_i > 0$, $= 0$ or $< 0$, respectively.
For a set or an event $A$, we use $\mathbb{I}\{A\}$ to denote the indicator function of $A$. For a set $A$, let $|A|$ be its cardinality.
For a matrix $\bM$, we denote by $\|\bM\|_{\max}=\max\limits_{i,j}|M_{ij}|$ its max norm, and by $\|\bM\|_q$ its induced $q$-norm for $q\in \mathbb{N}^{\star}\cup\{\infty\}$. For $\bM\in\R^{n\times m}$, $I\subseteq[n]$ and $J\subseteq [m]$, define $\bM_{IJ}=(\bM_{ij})_{i\in I,j\in J}$, $\bM_{I\cdot}=(\bM_{ij})_{i\in I,j\in [m]}$ and $\bM_{\cdot J}=(\bM_{ij})_{i\in [n],j\in J}$. For a vector $\bgamma\in\R^m$, define $\bgamma^{\otimes 0} = 1$, $\bgamma^{\otimes 1} = \bgamma$, $\bgamma^{\otimes 2} = \bgamma \bgamma^{\top}$, and $\bgamma_S=(\gamma_{i})_{i\in S}$ when $S\subseteq [m]$.
Let $\nabla$ and $\nabla^2$ be the gradient and Hessian operators. For $f:\R^p\rightarrow\R$, $x\in \R^p$ and $I,J\subseteq[p]$, define $\nabla_I f(x)=(\nabla f(x))_I$ and $\nabla^2_{IJ}f(x)=(\nabla^2 f(x))_{IJ}$. $\mathcal{N}(\bmu,\bSigma)$ refers to the normal distribution with mean vector $\bmu$ and covariance matrix $\bSigma$. For two numbers $a$ and $b$, $a \lor b$ and $a \land b$ denote their maximum and minimum, respectively.


\section{Problem Setup}\label{sec:setup}
\subsection{Cox's proportional hazards model}\label{sec:M_estimator}
Let $T$, $C$ and $\{\bx(t) \in \R^{p} : 0\leq t\leq \tau\}$ denote the survival time, censoring time and predictable covariate process, respectively, where $\tau < \infty$ is the study ending time. For each sample, only one of the survival and censoring times is observed, whichever happens first. Let $Z = T \land C$ be the observed time and $\delta = \mathbb{I}\{T \leq C\}$ be the censoring indicator. $T$ and $C$ are assumed to be conditionally independent given the covariates $\{\bx(t) : 0\leq t\leq \tau\}$. The observed data is an independent and identically distributed (i.i.d.) sample $\{ (\{\bx_i(t) : 0\leq t\leq \tau\},Z_i,\delta_i) \}_{i\in [n]}$ from the population $(\{\bx(t) : 0\leq t\leq \tau\}, Z, \delta)$, and for simplicity it is assumed that there are no tied observations and that the covariates are centered.

Cox's proportional hazards model \citep{Cox_1972, Cox_1975} is a semi-parametric model widely used to model time-to-event outcomes. In this model, the conditional hazard function \mbox{$\lambda(t\mid \bx(t))$} of the survival time $T$ at time $t$ given the covariate vector $\bx(t) \in \R^p$ is given by
\begin{equation}\label{cox_lambda}
\lambda(t\mid \bx(t)) = \lambda_0(t) \exp(\bx(t)^{\top}\bbeta^{\star}),
\end{equation}
where $\lambda_0(\cdot)$ is a baseline hazard function and $\bbeta^{\star}=(\beta_1^{\star},\dots,\beta_{p}^{\star})^{\top} \in \R^p$. The function $\lambda_0(\cdot)$ is unspecified in this semi-parametric model: it is a nuisance function, and $\bbeta^{\star}$ is the parameter vector of interest which is assumed to be sparse.

Let $\bX(t) \in \R^{n \times p}$ be the design matrix at time $0\leq t\leq \tau$, and $\bX = \{\bX(t) : 0\leq t\leq \tau\}$.
Let $N$ be the number of failures (satisfying $\delta = 1$) and $t_1 < \dots < t_N$ be the ordered failure times. For $j \in [N]$, let $(j)$ denote the label of the sample failing at time $t_j$, i.e., the individual with $j^{th}$ shortest survival time. The risk set $R_j = \{ i : Z_i \geq t_j \}$ at time $t_j$ is the set of samples still at risk at $t_j$. Cox's log-partial likelihood is given by
\begin{equation*}
Q(\bbeta;\bX,\bZ,\bdelta) = \sum_{j=1}^N \Big\{ \bx_{(j)}(t_j)^{\top} \bbeta - \log\big(\sum_{i \in R_j} \exp(\bx_i(t_j)^{\top} \bbeta)\big) \Big\}.
\end{equation*}
Define the loss $\cL(\bbeta;\bX,\bZ,\bdelta) = - n^{-1}Q(\bbeta;\bX,\bZ,\bdelta)$.
Note that $\cL$ depends on $\bX$ and $\bbeta$ only through the entries of the product $\bX \bbeta$.
Throughout the paper, we will use the notation $\cL(\bX \bbeta)$ and its gradient and Hessian matrix will be taken with respect to $\bbeta$. Keeping the design matrix $\bX$ in the notation will be useful as soon as we introduce factor modeling.

\subsection{Counting process formulation}\label{sec:counting}

We adopt the counting process formulation of \citet{Andersen_Gill_1982}. For $i \in [n]$ and $t \in [0, \tau]$, define the counting process $N_i(t) = \mathbb{I}\{Z_i \leq t,\delta_i = 1\}$ and the at-risk indicator process $Y_i(t) = \mathbb{I}\{ Z_i \geq t \}$. Let $\overline{N}(t) = \sum_{i=1}^n N_i(t)$.
Using this notation, the loss $\cL$ is given by
\begin{equation*}
\cL(\bX \bbeta) = -\frac{1}{n} \sum_{i=1}^n \int_0^{\tau} \{ \bx_i(t)^{\top} \bbeta \} d N_i(t) + \frac{1}{n} \int_0^{\tau} \log\left[\sum_{i=1}^n Y_i(t) \exp(\bx_i(t)^{\top} \bbeta)\right] d \overline{N}(t).
\end{equation*}
For $\ell \in \{0,1,2\}$, define the following quantities
\begin{equation*}
S^{(\ell)}(\bX,\bbeta,t) = \frac{1}{n} \sum_{i=1}^n Y_i(t) \{\bx_i(t)\}^{\otimes \ell} \exp(\bx_i(t)^{\top} \bbeta), \qquad s^{(\ell)}_{\bx}(\bbeta,t) = \E [ S^{(\ell)}(\bX,\bbeta,t) ].
\end{equation*}
To simplify future notation, also define the following
\begin{align}\label{V}
\bV(\bX,\bbeta,t)
& = \frac{S^{(2)}(\bX,\bbeta,t)}{S^{(0)}(\bX,\bbeta,t)} - \left[\frac{S^{(1)}(\bX,\bbeta,t)}{S^{(0)}(\bX,\bbeta,t)}\right]^{\otimes 2},\cr \noalign{\vskip10pt}
\bv_{\bx}(\bbeta,t)
& = \frac{s^{(2)}_{\bx}(\bbeta,t)}{s^{(0)}_{\bx}(\bbeta,t)} - \left[\frac{s^{(1)}_{\bx}(\bbeta,t)}{s^{(0)}_{\bx}(\bbeta,t)}\right]^{\otimes 2}.
\end{align}
We can write $\cL(\bX \bbeta) = -\frac{1}{n} \sum_{i=1}^n \int_0^{\tau} \{ \bx_i(t)^{\top} \bbeta \} d N_i(t) + \frac{1}{n} \int_0^{\tau} \log(n S^{(0)}(\bX,\bbeta,t)) d \overline{N}(t)$.
Treating $\cL(\bX \bbeta)$ as a function of $\bbeta$, its gradient and its Hessian matrix with respect to $\bbeta$ can be written as $\nabla \cL(\bX \bbeta) = -\frac{1}{n} \sum_{i=1}^n \int_0^{\tau} \left\{ \bx_i(t) - \frac{S^{(1)}(\bX,\bbeta,t)}{S^{(0)}(\bX,\bbeta,t)} \right\} d N_i(t)$, and $\nabla^2 \cL(\bX \bbeta) = \frac{1}{n} \int_0^{\tau} \bV(\bX,\bbeta,t)d\overline{N}(t)$, respectively.

The counting process $N_i(t)$ has intensity process $\lambda_i(t,\bbeta^{\star}) = \lambda_0(t) Y_i(t) \exp(\bx_i(t)^{\top} \bbeta^{\star})$, which does not admit jumps at the same time as $N_j(t)$ for $j \neq i$. Define the compensator $\Lambda_i(t) = \int_0^t \lambda_i(u,\bbeta^{\star}) du$, and $M_i(t) = N_i(t) - \Lambda_i(t)$. Then $M_i(t)$ is an orthogonal square-integrable local martingale with respect to the filtration $\calF_{t,i} = \sigma\{ N_i(u),\bx_i(u^+),Y_i(u^+) : 0 \leq u \leq t \}$.
Let $\calF_t = \bigcup_{\,i=1}^{\,n} \calF_{t,i}$ be the smallest $\sigma$-algebra containing all $\calF_{t,i}$. Then $\overline{M}(t) = \sum_{i=1}^n M_i(t)$ is a martingale with respect to $\calF_t$.

\subsection{Approximate factor model}\label{subsec-approx-factor-model}

Let $p_1$ be the number of time-dependent covariates and $p_2$ be the number of time-independ\-ent covariates, with $p_1+p_2=p$. Any covariate vector $\bx_i(t)\in \R^p$ can be decomposed into a time-dependent part $\bx^{(1)}_i(t)\in \R^{p_1}$ and a time-independent part $\bx^{(2)}_i\in \R^{p_2}$. Without loss of generality, the first $p_1$ columns of $\bX(t)$ correspond to the time-dependent covariates, and the following columns to the time-independent ones. We use $\bX_1(t)$ and $\bX_2$ to denote the submatrix of $\bX(t)$ corresponding to the time-dependent covariates and time-independent ones, respectively, and we define $\bX_1 = \{\bX_1(t) : 0\leq t\leq \tau\}$. Typically, $\bX_1$ can represent low-dimensional clinical data, and $\bX_2$ can constitute high-dimensional genomics or genetic data.

Assume that latent common factors drive the time-independent covariates.
We refer to \citet{Fan_etal_2021} for an overview of factor modeling, and to \citet{Lawley_Maxwell_1971, Stock_Watson_2002, Bai_Ng_2002, Fan_etal_2013} for specific settings.
The time-independent covariates $\{ \bx^{(2)}_i \}_{i\in [n]} \subseteq \R^{p_2}$ follow the approximate factor model
\begin{equation}\label{eq:factor}
\bx^{(2)}_i = \bB \bff_i + \bu_i,\qquad i\in[n],
\end{equation}
where $\{ \bff_i \}_{i\in [n]} \subseteq \R^K$ are latent factors with $\E[\bff_i] = 0$, $\bB\in\R^{p_2\times K}$ is a loading matrix, and $\{ \bu_i \}_{i\in [n]} \subseteq \R^{p_2}$ are idiosyncratic components with $\E[\bu_i] = 0$ and uncorrelated with the latent factors. The quantities $\bff_i$, $\bu_i$ and $\bB$ are not observable, they will be estimated from $\{ \bx^{(2)}_i \}_{i\in [n]}$.
Denote $\bF=(\bff_1, \dots, \bff_n)^{\top}\in\R^{n \times K}$ and $\bU=(\bu_1, \dots, \bu_n)^{\top} \in \R^{n \times p_2}$,
so that we can write~\eqref{eq:factor} as
\begin{equation}\label{factor-x}
\bX_2=\bF \bB^{\top} +\bU.
\end{equation}
One could be tempted to assume that the time-dependent covariates $\{\bx^{(1)}_i(t)\}_{i\in [n]}$ also follow an approximate factor model. However, as time $t$ increases, the number of samples still at risk (satisfying $Y_i(t)=1$) decreases, hence the estimation of latent factors would become less accurate for the time-dependent covariates. In addition, as the dimension $p_1$ of the time-dependent part is not large, the benefits of factor modeling would be limited.


\section{FarmHazard and Preliminary Theoretical Results}\label{sec:pre-theory}

\subsection{FarmHazard: definition and motivation}
In the modern big data environment where an enormous volume of information is accessible, it is possible to accurately estimate underlying latent factors and idiosyncratic components. Building upon latent factors, we propose a model named Factor-Augmented Regularized Model for Hazard Regression (FarmHazard).
Our latent model extends~\eqref{cox_lambda} and is defined by a conditional hazard function of the form
\begin{equation}\label{cox_augmented}
\lambda(t\mid \bx^{(1)}(t), \bff, \bu) = \lambda_0(t) \exp(\bx^{(1)}(t)^{\top} \bbeta_1^{\star} + \bff^{\top} \bgamma_2^{\star} + \bu^{\top} \bbeta_2^{\star}),
\end{equation}
where $\bbeta_1^{\star}$ and $\bbeta_2^{\star}$ are sparse vectors, and $\bx^{(1)}(t)$, $\bff$ and $\bu$ are defined in Section~\ref{subsec-approx-factor-model}.
It includes Principal Component Regression (PCR) on the time-independent covariates and the Cox model~\eqref{cox_lambda} as special cases, and provides a good variable selection procedure in the latter case because variables are decorrelated.
Indeed, if $\bbeta_2^{\star}=0$ in model~\eqref{cox_augmented}, we recover PCR on the time-independent predictors. On the other hand, if $\bgamma_2^{\star} = \bB^{\top} \bbeta_2^{\star}$, where $\bB$ is the loading matrix in~\eqref{factor-x}, we recover the usual Cox model under the approximate factor model. In that case, the sparse model in~\eqref{cox_lambda} translates to the sparse model in FarmHazard regression~\eqref{cox_augmented} and the latent factors $\bff$ are used for dependence adjustments. Note that the space spanned by $(\bx^{(2)}, \bff)$ is the same as that by $(\bu, \bff)$. This expands the space spanned by $\bx^{(2)}$ into the powerful principal component directions.

The penalized log-partial likelihood problem corresponding to~\eqref{cox_augmented} is
\begin{equation}\label{Profile_likelihood}
\min\limits_{\bbeta_1 \in \R^{p_1}, \ \bbeta_2 \in \R^{p_2}, \ \bgamma_2 \in \R^K}
\left\{
\cL(\bX_1 \bbeta_1 + \bF \bgamma_2 + \bU \bbeta_2)+\cP_\lambda((\bbeta_1^{\top},\bbeta_2^{\top})^{\top})) \right\},
\end{equation}
where $\cP_\lambda(\cdot)$ is a sparsity-inducing penalty function with regularization parameter $\lambda > 0$.
In the general case,~\eqref{Profile_likelihood} can be interpreted as a penalized factor-augmented hazard regression, where $\{(\bx^{(1)}_i(t)^{\top},\bu_i^{\top},\bff_i^{\top})^{\top}\}_{i\in [n]}$ are the covariates. By lifting the space of time-independent covariates from $\R^{p_2}$ to $\R^{p_2+K}$, the strongly correlated $\bx^{(2)}_i$ are replaced by the weakly correlated $(\bu_i, \bff_i)$, as the common dependent part in $\bu_i$ has already been taken out.
Thus,~\eqref{Profile_likelihood} removes the effect of strong correlations caused by the latent factors.
By calculating the solution of~\eqref{Profile_likelihood} and taking  its first $p_1+p_2=p$ entries, we considerably improve the accuracy of model selection in Cox's model. Note that
$\bU$ and $\bF$ are not observable, hence we need to plug in estimators $\hat{\bU}$ and $\hat{\bF}$ in~\eqref{Profile_likelihood}, and consider the estimator
\begin{equation*}
\hat{\btheta} = \argmin\limits_{\bbeta_1 \in \R^{p_1}, \ \bbeta_2 \in \R^{p_2}, \ \bgamma_2 \in \R^K}
\left\{
\cL(\bX_1 \bbeta_1 + \hat{\bF} \bgamma_2 + \hat{\bU} \bbeta_2)+\cP_\lambda((\bbeta_1^{\top},\bbeta_2^{\top})^{\top})) \right\}.
\end{equation*}
The estimators $\hat{\bU}$ and $\hat{\bF}$ are introduced in the next section.

\subsection{A two-step procedure}\label{sec:procedure}

In \citet{Fan_etal_2020} for generalized linear models, the authors proposed to use the sample covariance matrix, its leading eigenvalues and eigenvectors for factor estimation. We suggest a more general approach, described in the following two steps.
\begin{enumerate}
\item[(1)] Let $\bX_2 \in \R^{n\times p_2}$ be the design matrix corresponding to the time-independent covariates. Fit the approximate factor model~\eqref{factor-x} and denote by $\hat{\bB}$, $\hat{\bF}$ and $\hat{\bU}=\bX_2-\hat{\bF}\hat{\bB}^{\top}$ the obtained estimators of $\bB$, $\bF$ and $\bU$, respectively, by using principal component analysis \citep[e.g.][]{Bai_2003, Fan_etal_2013, Fan_etal_2018, SFDS, Fan_etal_2021}.
More specifically, let $\hat{\bSigma}_2$, $\hat{\bLambda} = \diag(\hat{\lambda}_1,\dots,\hat{\lambda}_K)$ and $\hat{\bGamma} = (\hat{\bxi}_1,\dots,\hat{\bxi}_K)$ be initial pilot estimators (not necessarily based on the sample covariance) for the covariance matrix $\bSigma_2$ of $\bx^{(2)}$, its leading $K$ eigenvalues $\bLambda = \diag(\lambda_1,\dots,\lambda_K)$ and their corresponding leading $K$ normalized eigenvectors $\bGamma = (\bxi_1,\dots,\bxi_K)$, respectively.
Compute $\hat{\bB}=(\hat{\lambda}_1^{1/2}\,\hat{\bxi}_1, \dots, \hat{\lambda}_K^{1/2}\,\hat{\bxi}_K)$ and $\hat{\bF} = \bX_2 \hat \bB \diag(\hat{\lambda}_1^{-1} \dots, \hat{\lambda}_K^{-1} )$. In general, the estimators $\hat{\bSigma}_2$, $\hat{\bLambda}$, and $\hat{\bGamma}$ can be constructed separately from different methods or even different sources of data. For sub-Gaussian distributions, one will choose $\hat{\bSigma}_2$ to be the sample covariance matrix, and $\hat{\bLambda}$, $\hat{\bGamma}$ to be its leading eigenvalues and eigenvectors. In this specific case, the above more general approach reduces to the following familiar solution: the columns of $\hat{\bF}/\sqrt{n}$ are the eigenvectors of $\bX_2\bX_2^{\top}$ corresponding to the leading $K$ eigenvalues and $\hat{\bB} = \bX_2^{\top} \hat{\bF} / n$. For heavy-tailed elliptical distributions, \citet{Fan_etal_2018} used robust estimators: the marginal Kendall's tau to obtain $\hat{\bSigma}_2$ and $\hat{\bLambda}$, and the spatial Kendall's tau to obtain $\hat{\bGamma}$.  See also \citet{Fan_etal_2019,Fan_etal_2021} for other robust covariance inputs such as elementwise truncated mean estimators.

\item[(2)] Define $\hat{\bW}_2=(\hat{\bU}, \hat{\bF}) \in \R^{n\times(p_2+K)}$, $\hat{\bW}(t)=(\bX_1(t), \hat{\bW}_2) \in \R^{n\times(p+K)}$ and
$\hat{\bW} = \{\hat{\bW}(t) : 0\leq t\leq \tau\}$. Solve the augmented problem
\begin{equation}\label{L1-ppmle}
\hat{\btheta} = \argmin\limits_{\btheta\in \R^{p+K}}
\left\{ \cL(\hat{\bW}\btheta)+\cP_\lambda(\btheta_{[p]}) \right\},
\end{equation}
and define $\hat{\bbeta} = \hat{\btheta}_{[p]}$ as the first $p$ entries.
\end{enumerate}
Let $\bW_2 = (\bU,\bF)$ be the unobservable augmented design matrix for the time-independent covariates. Let $\bW(t) = (\bX_1(t),\bW_2)$ and $\bW = \{\bW(t) : 0\leq t\leq \tau\}$. If $\bF$ and $\bU$ are well estimated (see Section~\ref{sec:factor_estimate}), the columns of $\hat\bW_2$ are weakly correlated.
The strongly dependent covariates $\bx^{(2)}_i$ are then replaced by the weakly dependent covariates $(\hat{\bu}_i^{\top},\hat{\bff}_i^{\,\top})^{\top}$.

\subsection{Properties in factor model estimation}\label{sec:factor_estimate}

We introduce the asymptotic properties of estimated factors $\hat{\bF}$ and idiosyncratic components $\hat{\bU}$ in Lemma~\ref{lem:factor} below.
To make~\eqref{eq:factor} identifiable, the following identifiability condition is usually imposed in the literature.
\begin{ass}[Identifiability]\label{assump:identifiability}
$\cov(\bff)=\bI_K$, and $\bB^{\top} \bB$ is diagonal.
\end{ass}
Under this condition, the covariance matrix of $\bx^{(2)}$ is $\bSigma_2 = \bB \bB^{\top} + \bSigma_u$, where $\bSigma_u$ is the covariance matrix of $\bu$.
Principal component analysis is used to recover the latent factors, the loading matrix and the idiosyncratic components. Its use can be proved to be fully justified under the usual assumption that the effect of the factors outweighs the noise. In order to quantify it, the following pervasiveness assumption is common in the literature.
\begin{ass}[Pervasiveness]\label{assump:pervasiveness}
All the eigenvalues of $\bB^{\top} \bB /p_2$ are bounded away from $0$ and $\infty$ as $p_2\to\infty$,
and $\|\bSigma_u\|_2$ is bounded.    
\end{ass}
The first part of the pervasiveness assumption holds for example if the factors loadings $\{\bb_j\}_{j\in [p_2]}$ are i.i.d.~realizations of a non-degenerate $K$-dimensional random vector with a finite second moment. The second part holds easily by a sparsity condition on $\bSigma_u$. We also make the following assumption.
\begin{ass}[Loadings and initial pilot estimators]\label{assump:pilot}
$\|\bB\|_{\max}$ is bounded, and $\hat{\bSigma}_2$, $\hat{\bLambda}$ and $\hat{\bGamma}$ satisfy $\|\hat{\bSigma}_2 - \bSigma_2\|_{\max} = O_{\PP}(\sqrt{(\log p_2)/n})$, $\|(\hat{\bLambda} - \bLambda) \hat{\bLambda}^{-1}\|_{\max} = O_{\PP}(\sqrt{(\log p_2)/n})$ and $\|\hat{\bGamma} - \bGamma\|_{\max} = O_{\PP}(\sqrt{(\log p_2)/(n p_2)})$.
\end{ass}
The first part of Assumption~\ref{assump:pilot} is common. The second part holds in many cases of interest, for example for the sample covariance matrix under sub-Gaussian distributions \citep{Fan_etal_2013}; it also holds for the marginal and spatial Kendall's tau estimators \citep{Fan_etal_2018} and the elementwise adaptive Huber estimator \citep{Fan_etal_2019}.

The following Lemma is a restatement of Theorem 10.4 and its corollaries in \cite{SFDS}.

\begin{lem}\label{lem:factor}
Suppose that Assumptions~\ref{assump:identifiability}--\ref{assump:pilot} hold. Then
\begin{align*}
& \max\limits_{i\in [n]}
\|\hat{\bff}_i - \bff_i\|_2
= O_{\mathbb{P}}\left(\sqrt{K/p_2}\left(\sqrt{(\log p_2)/n} + 1/\sqrt{p_2}\right) \max\limits_{i\in [n]} \|\bx^{(2)}_i\|_2 + \max\limits_{i\in [n]} \|\bB^{\top} \bu_i\|_2\, / p_2 \right),\cr
& \max\limits_{i\in [n]} \|\hat{\bu}_i - \bu_i\|_{\infty}
= O_{\mathbb{P}}\left(
\sqrt{K} \max\limits_{i\in [n]} \|\hat{\bff}_i - \bff_i\|_2
+ \left(\sqrt{(\log p_2)/n} + 1/\sqrt{p_2}\right) \max\limits_{i\in [n]} \|\bff_i\|_2
\right).
\end{align*}
In particular, if the random variables $\|\bx^{(2)}\|_2$ and $\|\bB^{\top} \bu\|_2$ are sub-Gaussian, then
\begin{equation*}
\max\limits_{i\in [n]} \|\hat{\bff}_i - \bff_i\|_2 = O_{\mathbb{P}}\left(K\sqrt{(\log n)/p_2} + K \sqrt{(\log p_2)(\log n)/n}\right).
\end{equation*}
If, additionally, $\|\bff\|_2$ is sub-Gaussian, then
\begin{equation*}
\max\limits_{i\in [n]} \|\hat{\bu}_i - \bu_i\|_{\infty} = O_{\mathbb{P}}\left(K^{3/2}\sqrt{(\log n)/p_2} + K^{3/2}\sqrt{(\log p_2)(\log n)/n}\right).
\end{equation*}
Furthermore, if $\|\bx^{(2)}\|_{\infty}$ and $\|\bff\|_{\infty}$ are bounded and $\|\bB^{\top} \bu\|_{\infty} \leq \mathcal{C} \sqrt{p_2}$ for some $\mathcal{C}>0$, both convergence rates are improved:
\begin{align*}
& \max\limits_{i\in [n]} \|\hat{\bff}_i - \bff_i\|_2 = O_{\mathbb{P}}\left( \sqrt{K(\log p_2)/n} + \sqrt{K/p_2} \right),\cr
\mbox{and} \quad
& \max\limits_{i\in [n]} \|\hat{\bu}_i - \bu_i\|_{\infty} = O_{\mathbb{P}}\left( K\sqrt{(\log p_2)/n} + K/\sqrt{p_2} \right).
\end{align*}
\end{lem}

\begin{rem}
Recall that $\bW_2 = (\bU,\bF)$ is the unobservable augmented design matrix for the time-independent covariates and $\hat{\bW}_2 = (\hat{\bU},\hat{\bF})$ is its estimator.
Using
\begin{align*}
\|\hat{\bW}_2 - \bW_2\|_{\max}
& = \big(\max\limits_{i\in [n]}
\|\hat{\bff}_i - \bff_i\|_{\infty}\big)
\lor
\big(\max\limits_{i\in [n]}
\|\hat{\bu}_i - \bu_i\|_{\infty}\big)
,
\end{align*}
we can bound $\|\hat{\bW}_2 - \bW_2\|_{\max}$ with Lemma~\ref{lem:factor}.
\end{rem}

In practical applications, the number $K$ of factors must be chosen before the estimation of factors, loading matrix, and idiosyncratic components.
Numerous methods exist to estimate $K$ \citep[e.g.][]{Bai_Ng_2002, Luo_etal_2009, Lam_Yao_2012, Ahn_Horenstein_2013, Chang_etal_2015, Fan_Guo_Zheng_2022}. We refer to \citet{SFDS} for an overview of these methods.
For the numerical experiments in Section~\ref{sec:num-exp}, we use the Adjusted Eigenvalues Thresholding (ACT) estimator \citep{Fan_Guo_Zheng_2022}, which works as follows. Let $\hat{\bR}_2 = \diag(\hat{\bSigma}_2)^{-1/2} \hat{\bSigma}_2 \diag(\hat{\bSigma}_2)^{-1/2}$ be the sample correlation matrix and $\{\lambda_j(\hat{\bR}_2)\}_{j\in [p_2]}$ be its eigenvalues.
For any $j\in [p_2]$, let $\lambda_j^C(\hat{\bR}_2)$ denote the bias-corrected estimator of the $j^{th}$ largest eigenvalue.
The ACT estimator is $\hat{K}
= |\{j : \lambda_j^C(\hat{\bR}_2) > 1 + \sqrt{p_2 / n}\}|$.
This method has the great advantage of being tuning-free and scale-invariant.

\subsection{Variable screening}\label{sec:screening}

In ultra-high dimensional problems where the dimension grows with the sample size in a non-polynomial fashion, regularized regression faces multiple statistical and computational challenges.
To remedy this, screening methods have been developed, but they tend to include too many variables when strong correlations exist among covariates. In this section, we introduce a factor-augmented variable screening procedure for Cox's proportional hazards model, which is able to overcome the correlation issues in ultra-high dimensional problems by using the weakly dependent factor-augmented new predictors in place of the original ones.
As in \citet{Zhao_Li_2012}, we assume for simplicity that all covariates are time-independ\-ent to carry out the screening analysis. In our notation, this means that we have $p_2=p$. The augmented covariates are standardized to have mean zero and standard deviation one. The procedure is as follows.
\begin{enumerate}
\item[(1)] Use the first step in Section~\ref{sec:procedure} to obtain $\hat{\bB}$, $\hat{\bF}$ and $\hat{\bU}$.

\item[(2)] For $j\in[p]$, let $\hat{\bU}_{\cdot [j]}$ be the $j$-th column of $\hat\bU$. Solve the augmented marginal regression
\begin{equation*}
(\hat\beta_j, \hat\bgamma_j) = \argmin\limits_{\beta \in \R, \bgamma \in \R^{K} }
\cL(\hat{\bU}_{\cdot [j]} \beta+ \hat{\bF} \bgamma).
\end{equation*}

\item[(3)] Return the set $\{ j : | \hat\beta_j | \geq \xi \}$ for some threshold $\xi$.
\end{enumerate}


\section{Theoretical Guarantees}\label{sec:theory}

To achieve sign and estimation consistencies, \citet{Zhao_Yu_2006} proposed an irrepresentable condition for $\ell_{1}$-penalized least-squares regression, and \citet{Lee_etal_2015} proposed a generalized version for general regularized $M$-estimators. When applied to some loss function $\cL$ with $\ell_{1}$ regularization, we can write it as
\begin{equation}\label{general_IC}
\|\nabla^2_{\cS^c \cS}\cL(\btheta^{\star}) [ \nabla^2_{\cS \cS}\cL(\btheta^{\star}) ]^{-1}\|_{\infty}\leq 1-\mu,
\end{equation}
for some $\mu\in(0,1)$, where $\btheta^{\star}$ is the true parameter, $\cS = \supp(\btheta^{\star})$ and $\cS^c$ is its complement.

When the covariates are highly correlated, this irrepresentable condition, also called mutual incoherence condition \citep{Lv_etal_2018} in the context of Cox's proportional hazards model, can easily fail. In this case, model selection consistency is very unlikely to be achieved. We can address this issue by applying the procedure described in Section~\ref{sec:procedure}. After the first step of this procedure, the covariates become weakly correlated, hence the condition~\eqref{general_IC} becomes more likely to hold with positive $\mu$ bounded away from zero. Consequently, the model selection consistency and estimation error bounds are improved.
As explained in Section~\ref{subsec-approx-factor-model}, typically $\bX_1(t)$ is low-dimensional and $\bX_2$ is high-dimensional. Hence, the collinearity is less likely in $\bX_1(t)$. Moreover, the covariates in $\bX_1(t)$ and $\bX_2$ are usually weakly correlated. For example, it seems safe to assume that time-dependent data such as age or number of cigarettes smoked daily do not exhibit high correlation with high-dimensional genomics or genetic data.

\subsection{Properties of FarmHazard}\label{sec:core}

In this section, we establish theoretical guarantees of the FarmHazard-L procedure, which solves~\eqref{L1-ppmle} for $\cP_{\lambda}(\bbeta) = \lambda \|\bbeta\|_{1}$ after estimation of the augmented time-independent data matrix $\bW_2$.
Let $\btheta^{\star}=
((\bbeta^{\star})^{\top},
(\bgamma^{\star}_2)^{\top})^{\top} \in \R^{p+K}$, where $\bgamma^{\star}_2 = \bB^{\top} \bbeta_2^{\star}$. Recall that $\bW(t) = (\bX_1(t),\bW_2)$ and $\bW = \{\bW(t) : 0\leq t\leq \tau\}$. Define $\cS = \supp(\btheta^{\star})$ and denote its complement by $\cS^c=[p+K]\backslash \cS$. Before presenting the consistency results, we introduce and discuss the following assumptions.

\begin{ass}[Restricted strong convexity]\label{assump:RSC-1}
There exist $\kappa_2>\kappa_{\infty}>0$ such that
	\begin{equation*}
	\|[ \nabla^2_{\cS \cS}\cL(\bW\btheta^{\star}) ]^{-1}\|_{q} \leq\frac{1}{4\kappa_{q}} \quad \mbox{for } q
	\in \{2, \infty\}.
	\end{equation*}
\end{ass}

\begin{ass}[Irrepresentable condition]\label{assump:IC}
There exists $\mu\in(0,1/2)$ such that
	\begin{equation*}
	\|\nabla^2_{\cS^{c} \cS} \cL(\bW\btheta^{\star}) [ \nabla^2_{\cS \cS} \cL(\bW\btheta^{\star}) ]^{-1}\|_{\infty}\leq 1-2\mu.
	\end{equation*}
\end{ass}

\begin{ass}[Factor model estimation] \label{assump:factor}
There exists $\varepsilon>0$ such that
\begin{equation*}
	\|\hat{\bW}_2-\bW_2\|_{\max}\leq\varepsilon, \quad \mbox{and} \quad
    \varepsilon \Big(2 \sup\limits_{t\in [0, \tau]} \|\bW(t)\|_{\max} + \varepsilon \Big) \leq \frac{\kappa_{\infty}\mu}{2 |\cS|},
\end{equation*}
where $\kappa_{\infty}$ and $\mu$ are defined in Assumptions~\ref{assump:RSC-1} and \ref{assump:IC}.
\end{ass}

As $\bgamma^{\star}_2$ is not penalized, we are particularly interested in $\cS_{\bbeta^{\star}} = \supp(\bbeta^{\star})$. Note that $|\cS_{\bbeta^{\star}}| \leq |\cS| \leq |\cS_{\bbeta^{\star}}|+K$. Assumptions~\ref{assump:RSC-1}--\ref{assump:factor} and Theorem~\ref{thm:consistency-estimated-factors} below could alternatively be written using this support set, at the expense of additional notation.

Assumptions~\ref{assump:RSC-1}--\ref{assump:factor} are deterministic. Conditional on these events, Theorem~\ref{thm:consistency-estimated-factors} below holds. When these events hold in high probability, so does Theorem~\ref{thm:consistency-estimated-factors}. These assumptions are quite mild.
Assumption~\ref{assump:RSC-1} involves only a small matrix and holds easily.
Assumption~\ref{assump:IC} holds with high probability as long as $\E[\int_0^{\tau} \bv_{\bw}(\btheta^{\star},t) dN(t)]$ satisfies a similar condition \citep{Lv_etal_2018}, where $\bv_{\bw}$ is similar to~\eqref{V} but in the notation of the augmented design matrix, and where the counting process $N$ has the same distribution as the $N_i$'s.
Assumption~\ref{assump:factor} holds with high probability thanks to Lemma~\ref{lem:factor}, and we can accommodate different settings for the covariates. For example, if the last set of assumptions in Lemma~\ref{lem:factor} holds and $K$ does not depend on $n$ (which is a frequent assumption in the factor model literature), then $|\cS| \left( \sqrt{(\log p_2)/n} + 1/\sqrt{p_2} \right) = o(1)$ is sufficient to ensure that Assumption~\ref{assump:factor} holds with high probability.

\begin{thm}\label{thm:consistency-estimated-factors}
	Suppose that Assumptions~\ref{assump:RSC-1}--\ref{assump:factor} hold. Let $\{\hat{\bw}_i(t)^{\top}\}_{i\in [n]}$ be the rows of $\hat{\bW}(t) = (\bX_1(t), \hat{\bW}_2)$. Define $M=6 |\cS|^{3/2} \sup\limits_{t\in [0, \tau]} \|\hat{\bW}(t)\|_{\max}^3$ and
\begin{align*}
\eta = \left\|\frac{1}{n} \sum_{i=1}^n \int_0^{\tau} \left\{ \hat{\bw}_i(t) - \frac{\sum_{j=1}^n Y_j(t) \hat{\bw}_j(t) \exp(\bx_j(t)^{\top} \bbeta^{\star})}{\sum_{j=1}^n Y_j(t) \exp(\bx_j(t)^{\top} \bbeta^{\star})} \right\} d N_i(t) \right\|_{\infty}.
\end{align*}
If $\frac{7}{\mu} \eta < \lambda < \frac{\kappa_2\kappa_{\infty}\mu}{12M\sqrt{|\cS|}}$, the solution of \eqref{L1-ppmle} with the $\ell_{1}$ penalty satisfies
$\supp(\hat{\bbeta})\subseteq \cS_{\bbeta^{\star}}$ and
\begin{eqnarray*}
\|\hat{\bbeta}-\bbeta^{\star}\|_{\infty}\leq \frac{6\lambda}{5\kappa_{\infty}}, \qquad  \|\hat{\bbeta}-\bbeta^{\star}\|_2\leq \frac{4\lambda\sqrt{|\cS|}}{\kappa_2},
\qquad \|\hat{\bbeta}-\bbeta^{\star}\|_1\leq \frac{6\lambda|\cS|}{5\kappa_{\infty}}.
\end{eqnarray*}
If there exists $C>7$ such that $\eta < \frac{\kappa_2\kappa_{\infty}\mu^2}{12CM\sqrt{|\cS|}}$ and
$\min\limits_{j \in \cS_{\bbeta^{\star}}}|\beta^{\star}_j| > \frac{6C}{5\kappa_{\infty}\mu}\eta$, then for $\frac{7}{\mu}\eta < \lambda < \frac{C}{\mu}\eta$, we have $\sign(\hat{\bbeta})=\sign(\bbeta^{\star})$.
\end{thm}
\begin{rem}
Theorem~\ref{thm:consistency-estimated-factors} is phrased in terms of $\hat{\bbeta}$, as $\bbeta^{\star}$ is the sparse vector of interest. A more general result on $\hat{\btheta}$ can be found in the Appendix.
\end{rem}

The choice $\lambda \asymp \eta$ in Theorem~\ref{thm:consistency-estimated-factors} ensures sign consistency, and the rates are
$\|\hat{\bbeta}-\bbeta^{\star}\|_{\infty}=O_{\mathbb{P}}(\eta)$, $\|\hat{\bbeta}-\bbeta^{\star}\|_2=O_{\mathbb{P}}(\eta \sqrt{|\cS|})$ and $\|\hat{\bbeta}-\bbeta^{\star}\|_1=O_{\mathbb{P}}(\eta |\cS|)$. We consequently need to control $\eta$, and the following lemma provides a probabilistic upper bound.

\begin{lem}\label{lem:gradient-upper-bound}
Suppose that Assumptions~\ref{assump:identifiability}--\ref{assump:pilot} hold. Then
\begin{equation*}
\eta = O_{\mathbb{P}}\left(
\|\hat{\bW}_2 - \bW_2\|_{\max} +
\sqrt{\frac{\log (p+K)}{n}} \sup\limits_{t\in [0, \tau]} \|\bW(t)\|_{\max}\right).
\end{equation*}
\end{lem}
\begin{rem}\label{rem-gradient-upper-bound}
In particular, if the last set of assumptions in Lemma~\ref{lem:factor} holds and $K$ does not depend on $n$, we can obtain $\eta = O_{\mathbb{P}}\left( \sqrt{(\log p)/n} + 1/\sqrt{p_2} \right)$.
\end{rem}

\subsection{Variable screening}

The factor-augmented variable screening procedure described in Section~\ref{sec:screening} is now studied in further detail. We fit marginal Cox regressions, which are potentially misspecified \citep{Struthers_Kalbfleisch_1986}.
Part of the ideas are based on \citet{Zhao_Li_2012} and are extended to incorporate factor modeling.
As in \citet{Zhao_Li_2012}, we assume for simplicity that all covariates are time-independent. In our notation, this means that we have $p_2=p$. The augmented covariates are standardized to have mean zero and standard deviation one. Recall from Section~\ref{sec:screening} that $\hat\beta_j$ and $\hat\bgamma_j$ are defined as
\begin{equation}\label{ppmle-screening}
(\hat\beta_j, \hat\bgamma_j) = \argmin\limits_{\bbeta \in \R, \bgamma \in \R^{K} }
\cL(\hat{\bU}_{\cdot [j]} \beta+ \hat{\bF} \bgamma).
\end{equation}

For $i \in [n]$ and $j \in [p]$, we let $\hat{\bw}_{ij} = (\hat{u}_{ij}, \hat{\bff}_i^{\,\top})^{\top} \in \R^{1+K}$, where $\hat{u}_{ij}$ is the entry of $\hat{\bU}_{\cdot [j]}$ corresponding to the $i^{th}$ sample and $\hat{\bff}_i^{\,\top}$ is the $i^{th}$ row of $\hat{\bF}$. Let $\bw_{ij} = (u_{ij}, \bff_i^{\top})^{\top}$ be the similar variable for the unobservable $\bU$ and $\bF$.
For $t\in [0, \tau]$, $\beta\in \R$ and $\bgamma\in \R^K$, define $S_j^{(\ell)}(\beta,\bgamma,t) = \frac{1}{n} \sum_{i=1}^n Y_i(t) \hat{\bw}_{ij}^{\otimes \ell} \exp(\hat{u}_{ij} \beta + \hat{\bff}_i^{\,\top} \bgamma)$, $s_j^{(\ell)}(\beta,\bgamma,t) = \E [ S_j^{(\ell)}(\beta,\bgamma,t) ]$,
$R_j^{(\ell)}(t) =
\frac{1}{n} \sum_{i=1}^n Y_i(t) \hat{\bw}_{ij}^{\otimes \ell} \lambda(t\mid \bx_i)$, and $r_j^{(\ell)}(t) = \E [ R_j^{(\ell)}(t) ]$.

For $j\in [p]$ and $(\beta,\bgamma) \in \R^{1+K}$, define the score
\begin{equation}\label{estim-equ}
\begin{split}
\bD_j(\beta,\bgamma)
& = -\frac{1}{n} \sum_{i=1}^n \int_0^{\tau} \left\{ \hat{\bw}_{ij} - \frac{S_j^{(1)}(\beta,\bgamma,t)}{S_j^{(0)}(\beta,\bgamma,t)} \right\} d N_i(t).
\end{split}
\end{equation}
and its population counterpart
\begin{equation*}
\begin{split}
\bd_j(\beta,\bgamma)
& = -\int_0^{\tau} \left\{ r_j^{(1)}(t) - \frac{s_j^{(1)}(\beta,\bgamma,t)}{s_j^{(0)}(\beta,\bgamma,t)} r_j^{(0)}(t)\right\} dt.
\end{split}
\end{equation*}
We then define $(\beta_j, \bgamma_j)$ as the solution of
\begin{equation}
\label{marginal-pop-version}
\bd_j(\beta_j,\bgamma_j) = \mathbf{0}_{1+K}.
\end{equation}
Under Assumptions~\ref{assump:SK1} and \ref{assump:SK2} below, which resemble Conditions $1$ and $2$ in \citet{Struthers_Kalbfleisch_1986}, and using their Theorem 2.1, we obtain that $(\hat\beta_j,\hat\bgamma_j)$ defined in~\eqref{ppmle-screening} is a consistent estimator of $(\beta_j,\bgamma_j)$. We consequently adopt these two assumptions in the rest of the section.
\begin{ass}[]\label{assump:SK1}
For each $j \in [p]$, there exists a neighborhood $\mathcal{B}_j$ of $(\beta_j,\bgamma_j)$ such that
\begin{equation*}
\sup_{t \in [0,
\tau], (\beta,\bgamma) \in \mathcal{B}_j} | S_j^{(0)}(\beta,\bgamma,t) - s_j^{(0)}(\beta,\bgamma,t) | \rightarrow 0 \quad \mbox{in probability as}\;\, n \rightarrow \infty,
\end{equation*}
$s_j^{(0)}(\beta,\bgamma,t)$ is bounded away from $0$ on $\mathcal{B}_j \times [0, \tau]$, and $s_j^{(0)}(\beta,\bgamma,t)$ and $s_j^{(1)}(\beta,\bgamma,t)$ are bounded on $\mathcal{B}_j \times [0, \tau]$.
\end{ass}

\begin{ass}[]\label{assump:SK2}
For each $j \in [p]$, $\int_0^{\tau} r_j^{(2)}(t)dt$ is finite.
\end{ass}

Let $F_T(\cdot\mid \bx)$ be the conditional cumulative distribution function of $T$, given the covariate vector $\bx$. Lemma~\ref{lem:screening-population} below gives a lower bound on the magnitude of the population marginal coefficients $\beta_j$.

\begin{ass}[]\label{assump:screening-bounded}
The idiosyncratic components $u_{ij}$ are bounded by a constant $M_0 > 0$.
\end{ass}

\begin{lem}\label{lem:screening-population}
Suppose that Assumptions~\ref{assump:SK1}--\ref{assump:screening-bounded} hold. Then $\beta_j$
defined in~\eqref{marginal-pop-version} satisfies
\begin{equation*}
\forall j \in \supp(\bbeta^{\star}),\, |\beta_j| \geq \frac{1}{2} M_0^{-2} | \cov (u_{1j}, \E[ F_T(C\mid \bx) \mid \bx ]) |.
\end{equation*}
\end{lem}

With the addition of the following two assumptions, which are similar to those in Section~\ref{sec:core}, the next theorem establishes the sure screening property.

\begin{ass}[Strong convexity of marginal loss functions]\label{assump:RSC-1-screening}
There exist $\kappa_2>\kappa_{\infty}>0$ such that for each $j \in [p]$,
	\begin{equation*}
	\|[ \nabla^2 \cL(\hat\bU_{\cdot [j]} \beta_j + \hat\bF \bgamma_j) ]^{-1}\|_{q} \leq\frac{1}{4\kappa_{q}} \quad \mbox{for } q
	\in \{2, \infty\}.   
	\end{equation*}
\end{ass}

\begin{ass}[Factor model estimation]\label{assump:factor-screening}
There exist constants $C_1,C_2>0$ such that \mbox{$\|\bW\|_{\max}\leq C_1$} and $\|\hat{\bW}-\bW\|_{\max}\leq C_2$.
\end{ass}

\begin{thm}\label{thm:screening}
Suppose that Assumptions~\ref{assump:SK1}--\ref{assump:factor-screening} hold with $C_2 \, (2 C_1 + C_2)$ small enough, and the assumptions in Remark~\ref{rem-gradient-upper-bound} are satisfied.

\noindent If $\displaystyle \xi \leq \nu \min_{j \in \supp(\bbeta^{\star})} \frac{1}{2} M_0^{-2} | \cov (u_{1j}, \E[ F_T(C\mid \bx) \mid \bx ]) |$ for some constant $\nu \in (0,1)$, and
\begin{align*}
\min_{j \in \supp(\bbeta^{\star})}
| \cov (u_{1j}, \E[ F_T(C\mid \bx) \mid \bx ]) | \gg \sqrt{(\log p)/n} + 1/\sqrt{p},
\end{align*}
then we obtain the sure screening property
\begin{equation*}
\PP (\supp(\bbeta^{\star}) \subseteq \{ j: |\hat\beta_j| \geq \xi \} ) \to 1.
\end{equation*}
\end{thm}


\section{Numerical Experiments}\label{sec:num-exp}

\subsection{Selection performance for high-dimensional data}\label{sec:num-exp-FarmHazard}

In this section, we present the results of simulations comparing the model selection performance of various procedures.
As described in Section~\ref{sec:procedure}, our FarmHazard algorithm starts with fitting an approximate factor model and then runs a regularized Cox regression with the augmented covariates. We use the package \texttt{glmnet}, which can handle optimization for the regularized Cox proportional hazards model \citep{Simon_etal_2011}.

We use the sample size $n=200$ and choose the dimension $p$ in the range $\left[1000, 5000\right]$. All experiments are based on $1000$ replications.
Consider a constant baseline hazard function $\lambda_0(t) = 1$. Therefore the conditional survival time $T\mid \bx$ is an exponential random variable with parameter $\exp(\bx^{\top} \bbeta^{\star})$ where $\bx$ is the covariate vector.
We independently set the distribution of the conditional censoring time $C\mid \bx$ to be exponential with parameter $\frac{3}{7} \exp(\bx^{\top} \bbeta^{\star})$. This ensures $30\%$ of the samples to be censored on average. The true coefficient vector $\bbeta^{\star}$ has $4$ non-zero entries, which are drawn uniformly at random in the interval $\left[2, 5\right]$.
The covariates $\bX$ are time-independent, generated from one of the following models.
\begin{enumerate}
\item[(1)] Factor model setting: we set $\bx_i=\bB \bff_i + \bu_i$ with $K = 3$ factors, where we sample the entries of $\bB$ and $\bff_i$ independently from $\mathcal{N}(0, 1)$ and $\bu_i$ from $\mathcal{N}(0, 2)$.
\item[(2)] Equicorrelated setting: we sample $\bx_i$ independently from $\mathcal{N}({\mathbf{0}_p}, \bSigma_\rho)$, where $\bSigma_\rho$ has diagonal elements 1 and off-diagonal elements $\rho$.  This is the one-factor model with equal factor loadings. Note that, as the $\bx_i$ are multivariate Gaussian, the covariates are independent when $\rho=0$.
\end{enumerate}
We compare the following five procedures. More details about these procedures can be found in the subsequent paragraphs.
\begin{enumerate}
\item[(1)] LASSO
\begin{equation*}
\hat{\bbeta}^{\,\text{LASSO}} = \argmin\limits_{\bbeta \in \R^{p}}
\left\{
\cL(\bX \bbeta)+\lambda \|\bbeta\|_1 \right\}.
\end{equation*}
\item[(2)] FarmHazard-L
\begin{equation*}
\hat{\btheta}^{\,\text{FarmH-L}} = \argmin\limits_{\bbeta \in \R^p, \ \bgamma \in \R^K}
\left\{
\cL(\hat{\bF} \bgamma + \hat{\bU} \bbeta)+\lambda \|\bbeta\|_1 \right\},
\; \textrm{and } \; \hat{\bbeta}^{\,\text{FarmH-L}} = \hat{\btheta}^{\,\text{FarmH-L}}_{[p]}.
\end{equation*}
\item[(3)] SCAD: one-step local linear approximation (LLA) of the SCAD penalty, with the LASSO initialization
\begin{equation*}
\hat{\bbeta} = \argmin\limits_{\bbeta \in \R^{p}}
\left\{
\cL(\bX \bbeta)+ \sum_{j=1}^p p'_{\lambda}(|\hbeta^{\,\text{LASSO}}_j|) |\beta_j| \right\},
\end{equation*}
where $p'_{\lambda}$ is the derivative of the SCAD penalty
and $\hat{\bbeta}^{\,\text{LASSO}}$ is the LASSO estimator.
\item[(4)] FarmHazard-S: factor-augmented hazard regression with one-step LLA of the SCAD penalty, with the FarmHazard-L initialization
\begin{equation}
\label{FarmHazard-S}
\hat{\btheta}^{\,\text{FarmH-S}} = \argmin\limits_{\bbeta \in \R^p, \ \bgamma \in \R^K}
\left\{
\cL(\hat{\bF} \bgamma + \hat{\bU} \bbeta)+\sum_{j=1}^p p'_{\lambda}(|\hbeta^{\,\text{FarmH-L}}_j|) |\beta_j| \right\},
\end{equation}
and $\hat{\bbeta}^{\,\text{FarmH-S}} = \hat{\btheta}^{\,\text{FarmH-S}}_{[p]}$, where $\hat{\bbeta}^{\,\text{FarmH-L}}$ is the FarmHazard-L estimator.
\item[(5)] Elastic-net
\begin{equation*}
\hat{\bbeta} = \argmin\limits_{\bbeta \in \R^{p}}
\left\{
\cL(\bX \bbeta)+\lambda( \alpha \|\bbeta\|_1 + (1-\alpha) \|\bbeta\|_2^2/2) \right\},
\; \textrm{with } \; \alpha = 0.9.
\end{equation*}
\end{enumerate}
The tuning parameter $\lambda$ of each procedure is computed by 10-fold sparse generalized cross-validation \citep{Bradic_etal_2011}. The model selection performance is measured by the sign consistency rate (that is, the empirical frequency of replications such that $\sign(\hat{\bbeta})=\sign(\bbeta^{\star})$) and the average size of the selected model (that is, of $\supp(\hat{\bbeta})$).
The average estimates of the selected model size are accompanied with $\pm \;2$ standard error intervals.
The surrounding confidence intervals for the sign consistency rate are $95\%$ Wilson intervals \citep{Wilson_1927}, which are known to provide more accurate coverage for binomial proportions than $\pm \;2$ standard error intervals \citep{Brown_etal_2001}.
The results are presented in Figure~\ref{fig_factor} for the factor model setting with $p$ in the range $\left[1000, 5000\right]$, and Figure~\ref{fig_correl} for the equicorrelated setting with $p = 2000$ and $\rho$ in the range $\left[0.0, 0.8\right]$.

The results show that FarmHazard-L and FarmHazard-S outperform LASSO \citep{Tibshirani_1996, Tibshirani_1997}, SCAD \citep{Fan_1997, Fan_Li_2001, Fan_Li_2002} and Elastic-net \citep{Zou_Hastie_2005} in all settings, for both measures of model selection performance.

The derivative of the smoothly clipped absolute deviation (SCAD) penalty is given by
$p'_{\lambda}(\beta) = \lambda \left( \mathbb{I}\{\beta \leq \lambda\} + \frac{(a \lambda - \beta)_{+}}{(a - 1) \lambda} \mathbb{I}\{\beta > \lambda\} \right)$
for $a = 3.7$.
More precisely, we use the one-step local linear approximation (LLA) \citep{Zou_Li_2008, Fan_etal_2014} of the SCAD penalty. This amounts to a weighted $\ell_{1}$ penalization procedure, with weights $(p'_{\lambda}(|\beta_j|))_{j\in [p]}$ for an initialization vector $\bbeta = (\beta_1, \dots, \beta_p)$. The most-used initialization is the LASSO estimator $\hat{\bbeta}^{\,\text{LASSO}}$, which corresponds to the procedure that we denote by SCAD above.  Note that LASSO itself can be regarded as the SCAD using zero as initialization and hence this procedure is also a two-step LLA procedure starting from the zero initialization.

We perform additional experiments in the factor model setting with various values for the dimension $p$ and the true vector $\bbeta^{\star}$. The results in Table \ref{tab_factor} show that FarmHazard greatly outperforms LASSO. The sign consistency rate of FarmHazard equals $1$ or is very close to $1$ and the selected model size is almost always equal to the true support size, while the sign consistency rate of LASSO is very low and its selected model size is too large.

\newcommand{\subfigfracinthree}{0.333}
\newcommand{\subfigfracintwo}{0.5}
\newcommand{\subfigfracin}{0.49}
\newcommand{\imgspace}{.01}
\newcommand{\imgfrac}{1}
\newcommand{\imgfracnotfull}{.65}
\newcommand{\imgfrachalf}{0.5}
\newcommand{\vgap}{.01}

\begin{figure}[H]
\centering
    \begin{subfigure}{\subfigfracin\linewidth}
       \includegraphics[width=\imgfrac\linewidth]{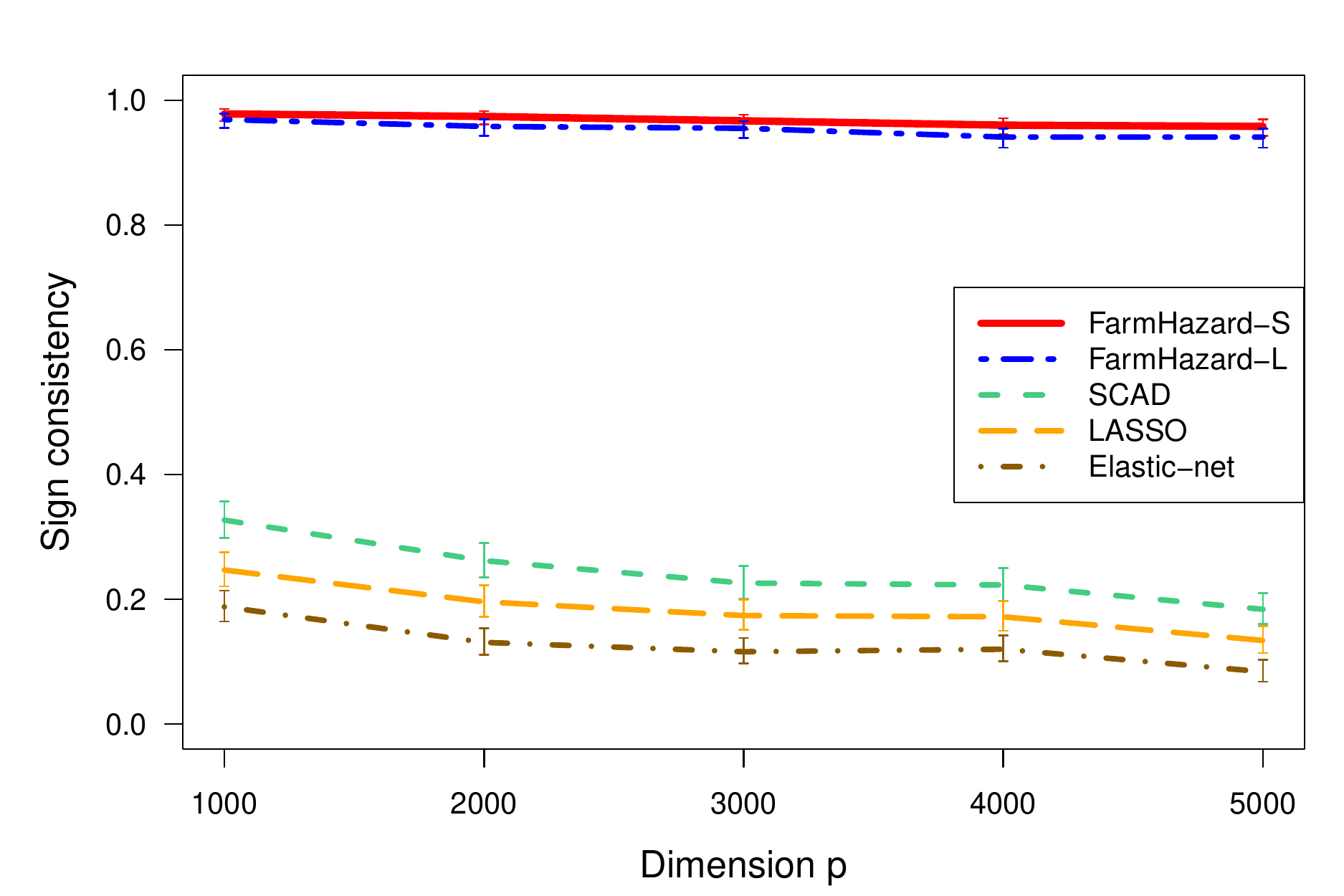}
    \end{subfigure}\hspace{\imgspace\linewidth}%
    \begin{subfigure}{\subfigfracin\linewidth}
     \centering
        \includegraphics[width=\imgfrac\linewidth]{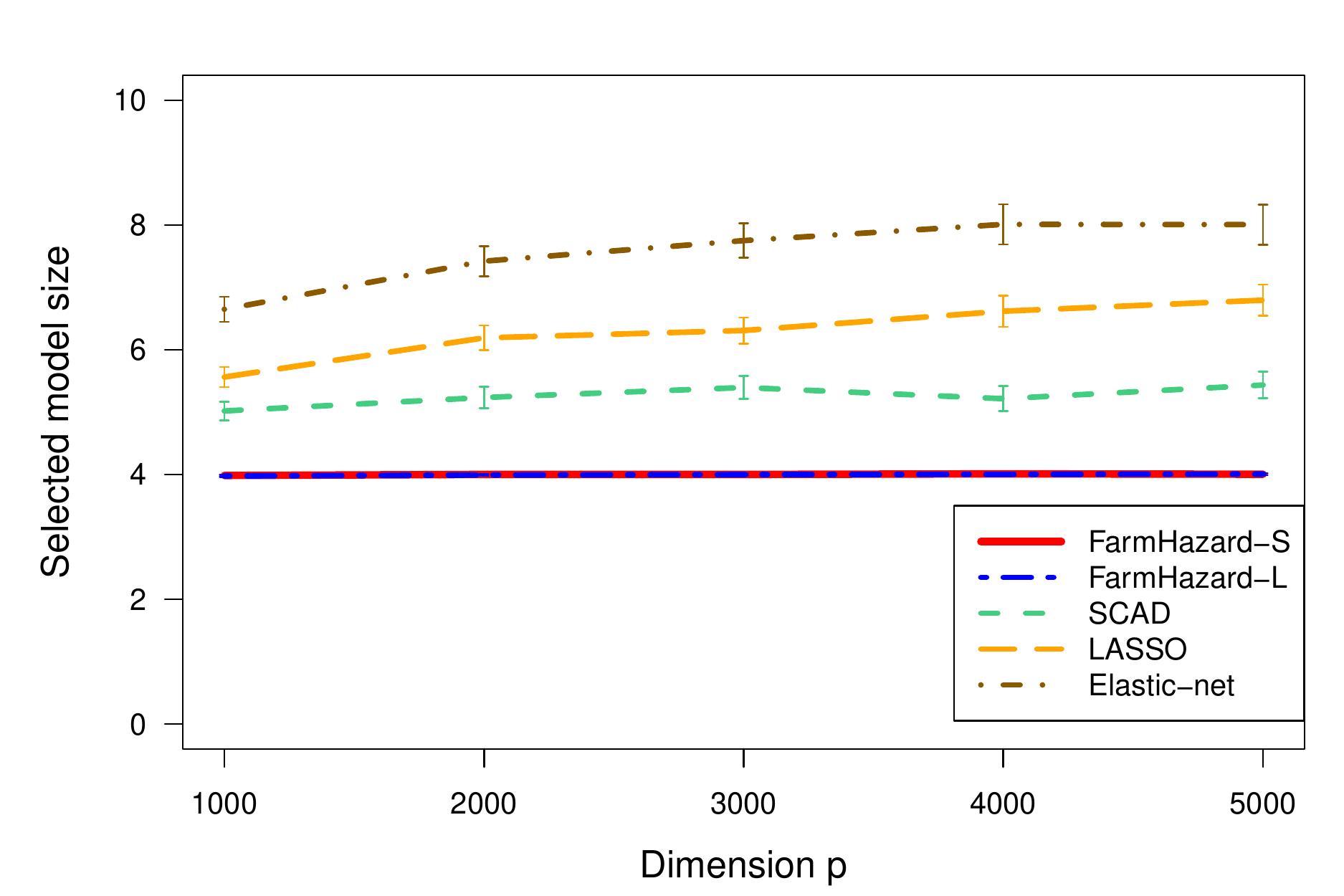}
    \end{subfigure}
    \caption{\footnotesize{Sign consistency rate (left) and selected model size (right) in the $3$-factor model setting.
}
}
    \label{fig_factor}
\end{figure}

\vspace*{-5mm}

\spacingset{1.6}
\begin{table}[H]
\scriptsize
\begin{center}
\caption{\footnotesize{
Results in the 3-factor model setting for various $p$ and $\bbeta^{\star}$.
$\bbeta^{\star}_1$ has $4$ non-zero entries, which are all $2$, and $\bbeta^{\star}_2$ has $3$ non-zero entries, which are drawn uniformly at random in $\left[0.5, 3\right]$.
Standard errors are in parentheses.
}}
\label{tab_factor}
\begin{tabular*}{\textwidth}{c@{\extracolsep{\fill}}*{6}{c}}
\hline\hline
\multicolumn{1}{c}{} & \multicolumn{2}{c}{FarmHazard-S} & \multicolumn{2}{c}{FarmHazard-L} & \multicolumn{2}{c}{LASSO}\\
\cline{2-3} \cline{4-5} \cline{6-7}
\multicolumn{7}{c}{Sign consistency rate}\\
{} & {$p =$ 500} & {$p =$ 5000} & {$p =$ 500} & {$p =$ 5000} & {$p =$ 500} & {$p =$ 5000}\\
\hline
{$\bbeta^{\star}_1$\phantom{000}} & 1.00 (0.00) & 0.97 (0.01) & 0.99 (0.003) & 0.96 (0.01) & 0.34 (0.01) & 0.15 (0.01)\\
{$\bbeta^{\star}_2$\phantom{000}} & 0.99 (0.003) & 0.97 (0.01) & 0.91 (0.01) & 0.90 (0.01) & 0.46 (0.02) & 0.25 (0.01)\\
\hline
\multicolumn{7}{c}{Selected model size}\\
{} & {$p =$ 500} & {$p =$ 5000} & {$p =$ 500} & {$p =$ 5000} & {$p =$ 500} & {$p =$ 5000}\\
\hline
{$\bbeta^{\star}_1$\phantom{000}} & 4.00 (0.00) & 4.03 (0.01) & 4.01 (0.003) & 4.04 (0.01) & 5.39 (0.06) & 7.38 (0.13)\\
{$\bbeta^{\star}_2$\phantom{000}} & 3.00 (0.003) & 3.02 (0.01) & 2.92 (0.01) & 2.92 (0.01) & 3.66 (0.04) & 4.93 (0.08)\\
\hline\hline
\end{tabular*}
\end{center}
\end{table}

\vspace*{-10mm}

\begin{figure}[H]
\centering
    \begin{subfigure}{\subfigfracin\linewidth}
       \includegraphics[width=\imgfrac\linewidth]{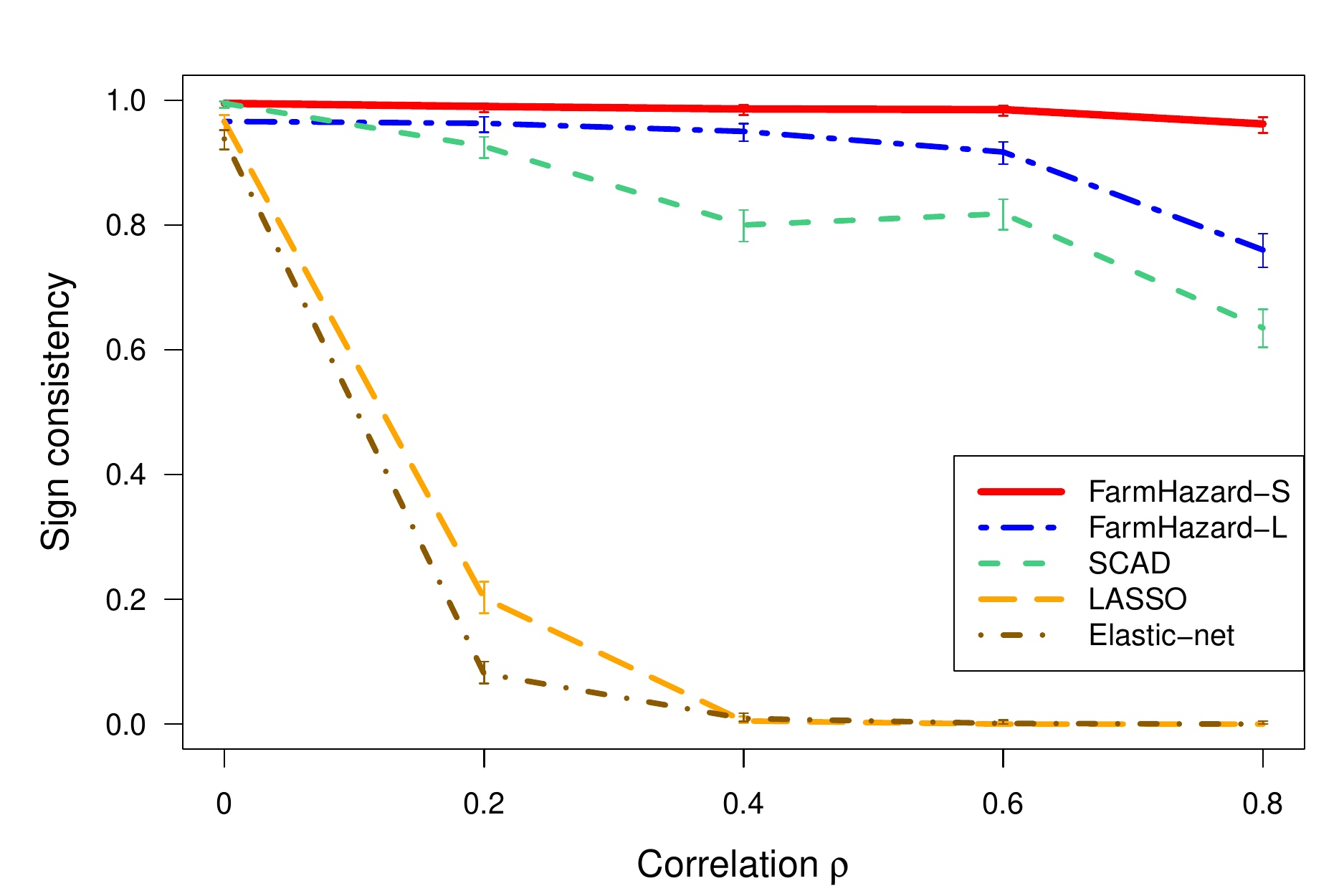}
    \end{subfigure}\hspace{\imgspace\linewidth}%
    \begin{subfigure}{\subfigfracin\linewidth}
     \centering
        \includegraphics[width=\imgfrac\linewidth]{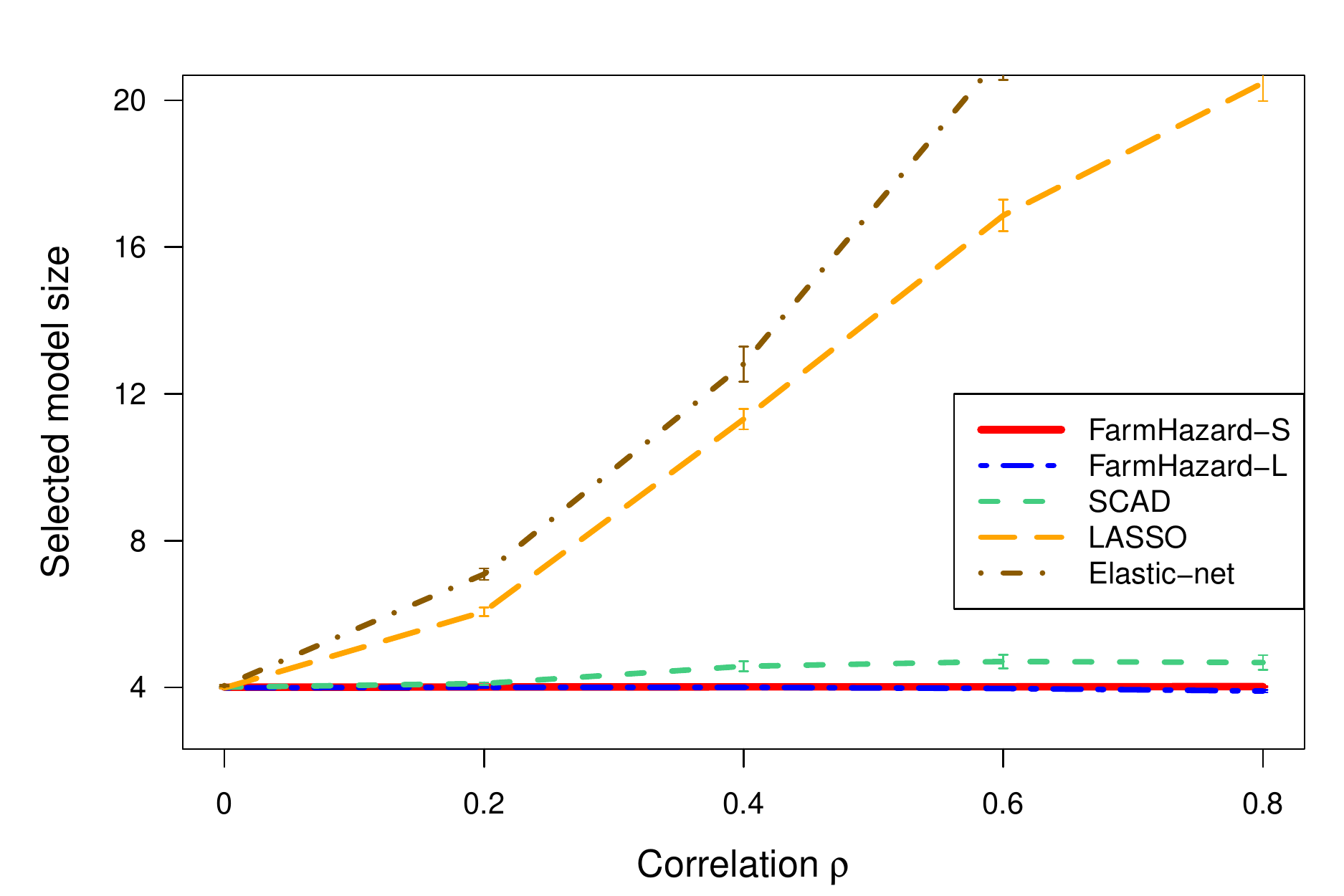}
    \end{subfigure}    
    \caption{\footnotesize{Sign consistency rate (left) and selected model size (right) in the equicorrelated setting, $p = 2000$.
}
}
    \label{fig_correl}
\end{figure}

\spacingset{1.5}

We also perform experiments in the equicorrelated setting with various values for the dimension $p$ and the correlation coefficient $\rho$. The results in Table \ref{tab_correlated} demonstrate that for any choice of dimension, FarmHazard greatly outperforms LASSO, and this becomes even more marked as the correlation $\rho$ increases. Regarding the uncorrelated setting ($\rho = 0$), FarmHazard-L and LASSO are similar, as expected, which in turn shows no price is paid by using FarmHazard-L.
For any dimension and correlation, the sign consistency rate of FarmHazard equals $1$ or is very close to $1$ and the selected model size is almost always equal to the true support size, while the sign consistency rate of LASSO decreases towards $0$ very quickly and its selected model size increases sharply when the correlation increases.

\spacingset{1.7}
\begin{table}[H]
\scriptsize
\begin{center}
\caption{\footnotesize{
Results in the equicorrelated setting for various $p$ and $\rho$.
Standard errors are in parentheses.
}}
\label{tab_correlated}
\begin{tabular*}{\textwidth}{c@{\extracolsep{\fill}}*{6}{c}}
\hline\hline
\multicolumn{1}{c}{} & \multicolumn{2}{c}{FarmHazard-S} & \multicolumn{2}{c}{FarmHazard-L} & \multicolumn{2}{c}{LASSO}\\
\cline{2-3} \cline{4-5} \cline{6-7}
\multicolumn{7}{c}{Sign consistency rate}\\
{} & {$p =$ 1000} & {$p =$ 3000} & {$p =$ 1000} & {$p =$ 3000} & {$p =$ 1000} & {$p =$ 3000}\\
\hline
{$\rho$ = 0.0\phantom{000}} & 1.00 (0.00) & 1.00 (0.00) & 0.98 (0.01) & 0.96 (0.01) & 0.98 (0.01) & 0.96 (0.01)\\
{$\rho$ = 0.4\phantom{000}} & 0.99 (0.003) & 0.99 (0.003) & 0.95 (0.01) & 0.93 (0.01) & 0.02 (0.004) & 0.01 (0.004)\\
{$\rho$ = 0.8\phantom{000}} & 0.97 (0.01) & 0.94 (0.01) & 0.78 (0.01) & 0.73 (0.01) & 0.00 (0.00) & 0.00 (0.00)\\
\hline
\multicolumn{7}{c}{Selected model size}\\
{} & {$p =$ 1000} & {$p =$ 3000} & {$p =$ 1000} & {$p =$ 3000} & {$p =$ 1000} & {$p =$ 3000}\\
\hline
{$\rho$ = 0.0\phantom{000}} & 4.00 (0.00) & 4.00 (0.00) & 3.99 (0.01) & 3.99 (0.01) & 3.99 (0.01) & 3.99 (0.01)\\
{$\rho$ = 0.4\phantom{000}} & 4.01 (0.002) & 4.01 (0.003) & 3.99 (0.01) & 4.01 (0.01) & 9.93 (0.11) & 11.78 (0.17)\\
{$\rho$ = 0.8\phantom{000}} & 4.01 (0.01) & 4.04 (0.01) & 3.83 (0.02) & 3.92 (0.02) & 18.93 (0.18) & 22.18 (0.29)\\
\hline\hline
\end{tabular*}
\end{center}
\end{table}

\spacingset{1.5}

\subsection{Variable screening}

We now illustrate the performance of our augmented variable screening procedure and compare it with sure independence screening for Cox's proportional hazards model \citep{Fan_etal_2010, Zhao_Li_2012}. The sample size is $n=200$, the dimension is $p = 10000$, and the true coefficient vector is $\bbeta^{\star} = (1,1,1,1, \mathbf{0}_{p-4}^{\top})^{\top}$. The covariates are $\bx_i=\bB \bff_i + \bu_i$ with $K = 3$ factors, where we sample the entries of $\bB$, $\bff_i$ and $\bu_i$ independently from $\mathcal{N}(0, 1)$, and the censoring mechanism is the same as in the previous section.
The screening procedures involve univariate regressions without penalization, therefore we use the package \texttt{survival} \citep{Therneau_2021}.
We run our experiments with several choices for the number of covariates selected. The screening performance is measured by the sure screening rate, the average false negative rate and the ROC curve. The surrounding confidence intervals are $95\%$ Wilson intervals for the sure screening rate, as it is a binomial proportion, and $\pm \;2$ standard error intervals for the false negative rate.
The average values and confidence intervals are computed over $1000$ replications, and Figure~\ref{fig_screening} displays the results. We experimentally confirm that the usual screening method performs poorly in this setting and that our factor-augmented variable screening remarkably improves upon it, yielding higher sure screening rate, lower false negative rate and better ROC curve.

\begin{figure}[H]
\centering
    \begin{subfigure}{\subfigfracinthree\linewidth}
    \centering
        \includegraphics[width=.9\linewidth]{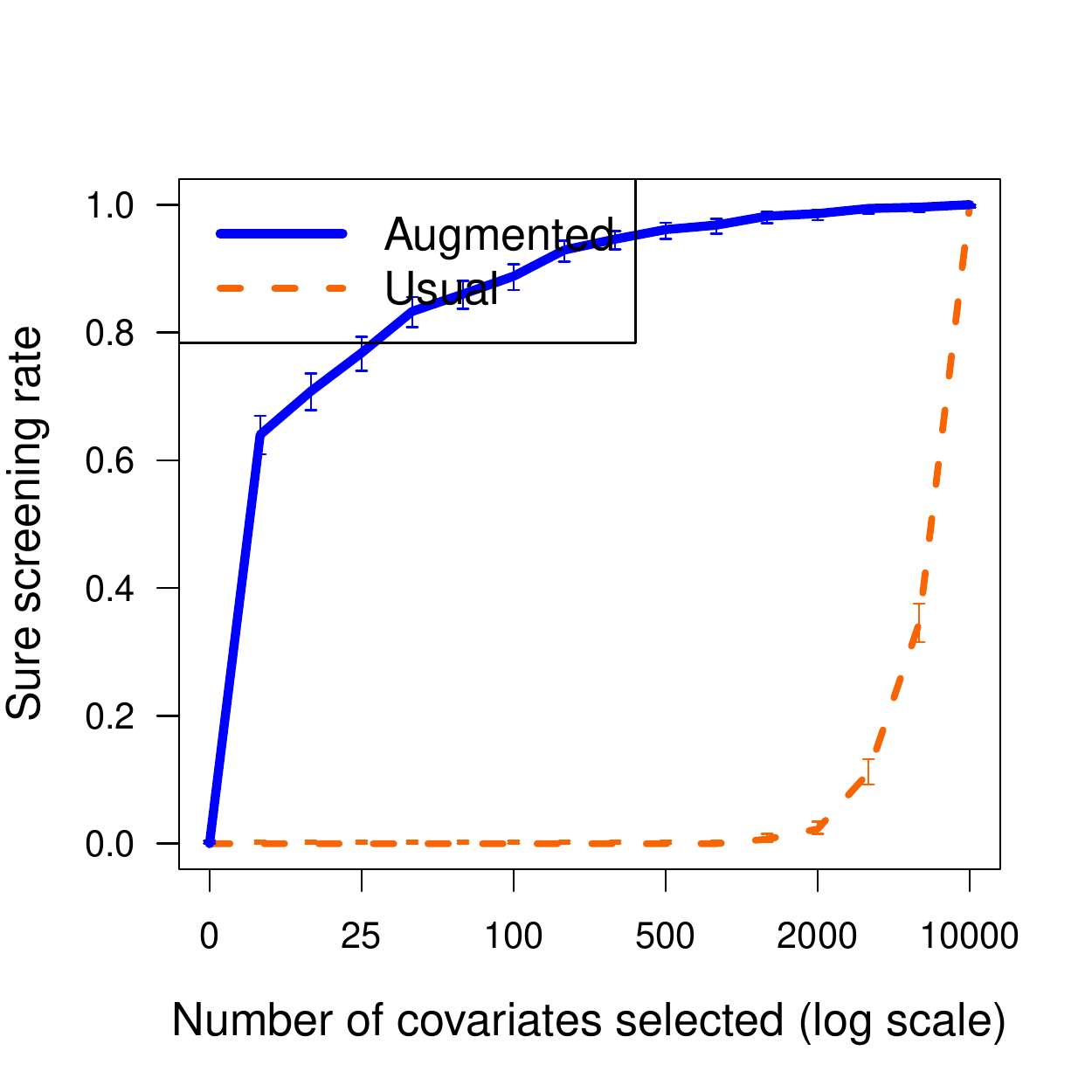}
    \end{subfigure}\hfill
     \begin{subfigure}{\subfigfracinthree\linewidth}
     \centering
         \includegraphics[width=.9 \linewidth]{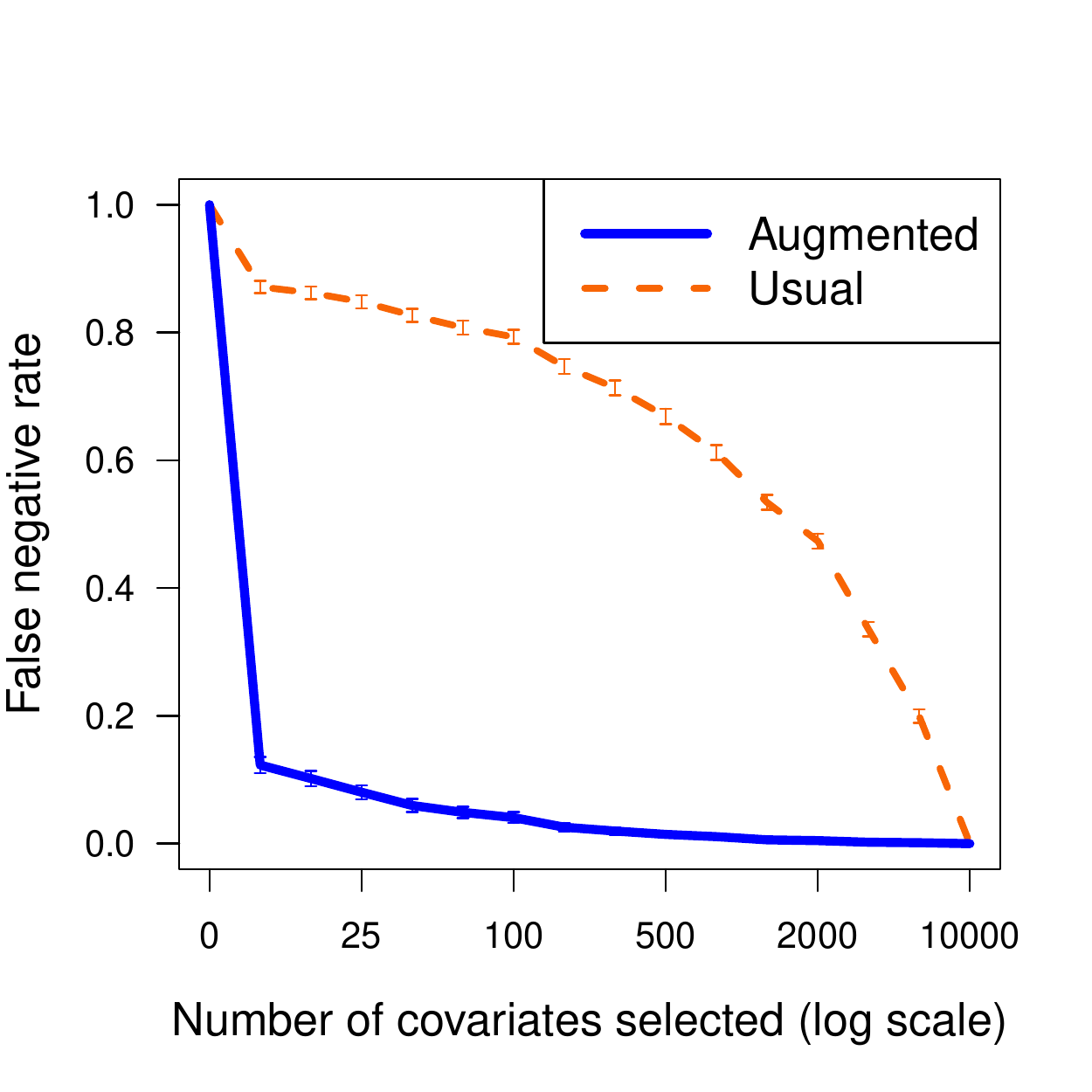}
    \end{subfigure}\hfill
    \begin{subfigure}{\subfigfracinthree\linewidth}
     \centering
         \includegraphics[width=.9 \linewidth]{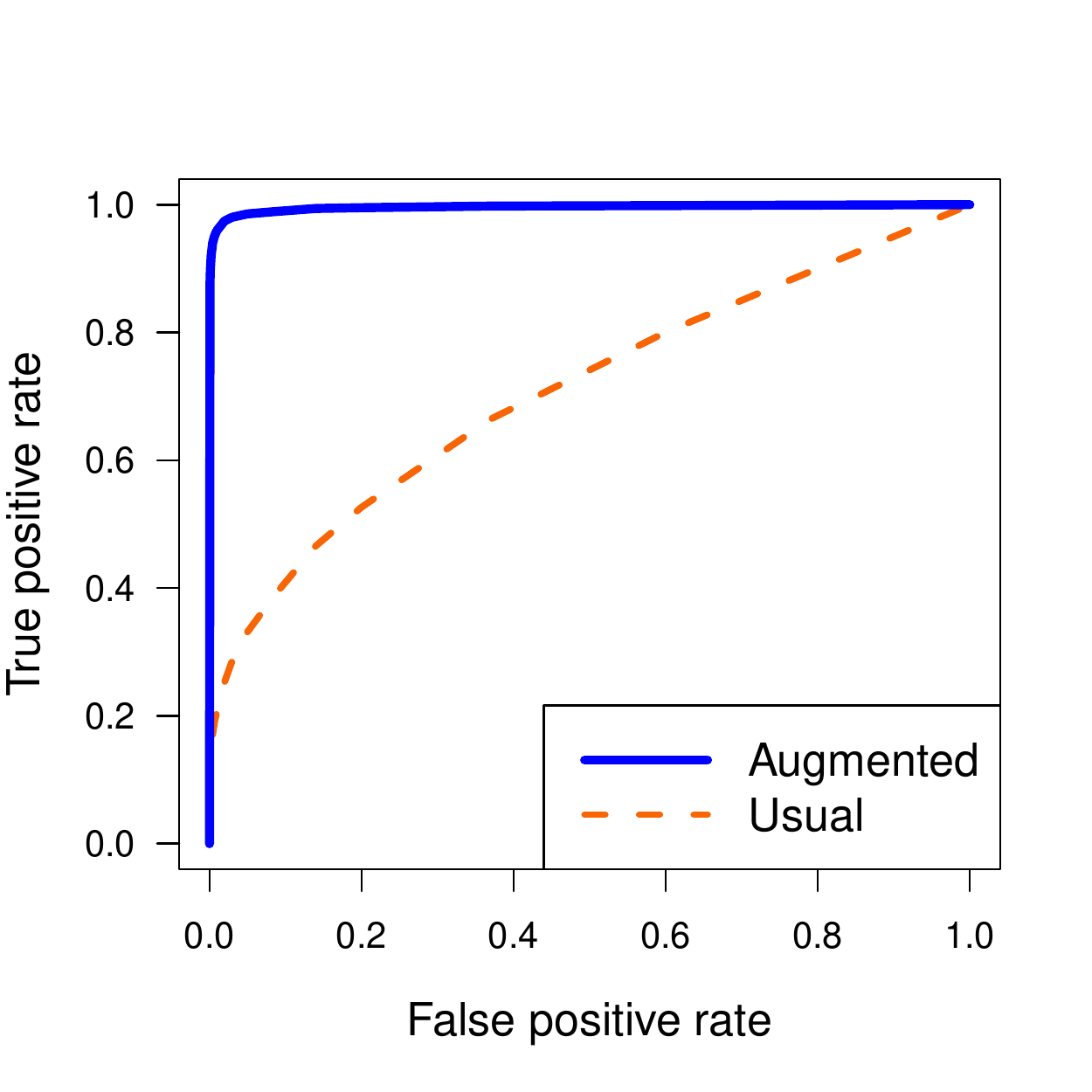}
    \end{subfigure}
    \caption{\footnotesize{Sure screening rate (left), false negative rate (middle) and ROC curve (right) of the augmented and usual procedures.
    }
    }
    \label{fig_screening}
\end{figure}

\subsection{Real data analysis}
We apply our procedure to the diffuse large-B-cell lymphoma (DLBCL) dataset\footnote{The dataset is available at \url{https://llmpp.nih.gov/DLBCL/}.} of \citet{Rosenwald_etal_2002}. Gene-expression profiles, with a total of $7399$ microarray features, are related to survival time. A sample of $240$ patients with untreated DLBCL is available in this study, of which $138$ died. We use median value imputation: the missing values for each covariate are replaced by the median of the observed values for this predictor, and we standardize the data. Five patients have follow-up time equal to zero; we remove them from the study. We screen the top $1500$ covariates out of the $7399$ present in the dataset, and hence reduce the dimensionality.
We then fit a regularized Cox model with FarmHazard-L, FarmHazard-S and LASSO to the data, the tuning parameters being computed by $10$-fold cross-validation. The genes selected by at least two procedures are presented in Table~\ref{tab_common_genes}.
The Gene column is a multi-field description of the Lymphochip microarray feature, with each field separated by a vertical bar.
The FarmHazard-L, FarmHazard-S and LASSO columns represent the estimated coefficients for these genes.
We note that the three procedures yield sign consistency of the estimated coefficients among the common selected covariates.
Additionally, we show the top $10$ coefficients of each procedure in Table~\ref{tab_top_genes}. In both tables, an X indicates a gene not selected by a procedure, i.e.~whose estimated coefficient is exactly zero. In Table~\ref{tab_top_genes}, an $<$ sign denotes a gene that does not appear in the top $10$ genes of a procedure but whose estimated coefficient is non-zero.

\spacingset{1.8}
\begin{table}[H]
\begin{center}
\caption{\small{Genes selected by at least two procedures and estimated coefficients}}
\label{tab_common_genes}
\scriptsize
\vspace{2mm}
\begin{tabular}{p{0.67\linewidth}rrr}
\hline\hline
{Gene}   &   {FarmH-L}   &   {FarmH-S}   &   {LASSO}
\\\cline{1-4}
|BC012161|*AA262133|Hs.99741|septin 1 & 0.162 & 0.219 & 0.186\\
|AF134159|*AA761323|Hs.15106|chromosome 14 open reading frame 1 & 0.105 & 0.081 & 0.096\\
||*AA808306|Hs.252627|ESTs & $-$0.080 & $-$0.085 & $-$0.135\\
||||LC\_26716 & 0.057 & 0.046 & 0.012\\
|D42043|*AA243583|Hs.79123|KIAA0084 protein & $-$0.048 & $-$0.052 & $-$0.207\\
|X53586|*R68760|Hs.227730|integrin, alpha 6 & $-$0.043 & $-$0.032 & $-$0.051\\
|M68956|$\sim$AA702254|Hs.75607|myristoylated alanine-rich protein kinase C substrate & 0.040 & 0.026 & 0.002\\
|D88532|*AI219836|Hs.372548|phosphoin ositide-3-kinase, regulatory subunit, polypeptide 3 (p55, gamma) & 0.004 & 0.015 & 0.110\\
||||LC\_32424 & 0.004 & X & 0.024\\
|U84143|~AA766589|Hs.180015|D-dopachrome tautomerase & X & 0.009 & 0.007\\
|M60527|*AA236906|Hs.709|deoxycytidine kinase & $-$0.241 & $-$0.380 & X\\
|S37431|*AI281565|Hs.181357|laminin receptor 1 (67kD, ribosomal protein SA) & $-$0.106 & $-$0.122 & X\\
||*AI392710|Hs.355401|Homo sapiens, clone IMAGE:4291796, mRNA, partial cds & 0.085 & 0.078 & X\\
|Y10659|*AA398286|Hs.285115|interleukin 13 receptor, alpha 1 & 0.063 & 0.057 & X\\
|X59812|*H98765|Hs.82568|cytochrome P450, subfamily XXVIIA (steroid 27-hydroxylase, cerebrotendinous xanthomatosis), polypeptide 1 & $-$0.047 & $-$0.065 & X\\
||*AI081246|Hs.122983|ESTs & $-$0.046 & $-$0.023 & X\\
|M24283|*R77293|Hs.168383|intercellular adhesion molecule 1 (CD54), human rhinovirus receptor & 0.045 & 0.040 & X\\
|AB020662|~AA054122|Hs.182982|golgin-67 & $-$0.043 & $-$0.073 & X\\
|BC000524||Hs.350166|ribosomal protein S6 & $-$0.036 & $-$0.022 & X\\
||*AI391477|Hs.343912|CAC-1 & 0.034 & 0.023 & X\\
|X66079|*AA490435|Hs.192861|Spi-B transcription factor (Spi-1/PU.1 related) & 0.033 & 0.053 & X\\
||||LC\_31500 & 0.023 & 0.024 & X\\
|AK023686||Hs.118183|hypothetical protein FLJ22833 & $-$0.021 & $-$0.037 & X\\
|U19970|*AA770349|Hs.51120|cathelicidin antimicrobial peptide & 0.021 & 0.033 & X\\
||||LC\_30576 & 0.018 & 0.007 & X\\
|M63438|$\sim$AA291844|Hs.156110|immunoglobulin kappa constant & $-$0.015 & $-$0.010 & X\\
||*AA805749|| & 0.014 & 0.019 & X\\
|AK001549|*AI492096|Hs.29379|hypothetical protein FLJ10687 & 0.013 & 0.028 & X
\\\hline\hline
\end{tabular}
\end{center}
\end{table}

\begin{table}[H]
\begin{center}
\caption{\small{Top $10$ coefficients of each procedure}}
\label{tab_top_genes}
\scriptsize
\vspace{2mm}
\begin{tabular}{p{0.67\linewidth}rrr}
\hline\hline
{Gene}   &   {FarmH-L}   &   {FarmH-S}   &   {LASSO}
\\\cline{1-4}
|M60527|*AA236906|Hs.709|deoxycytidine kinase & $-$0.241 & $-$0.380 & X\\
|BC012161|*AA262133|Hs.99741|septin 1 & 0.162 & 0.219 & 0.186\\
|S37431|*AI281565|Hs.181357|laminin receptor 1 (67kD, ribosomal protein SA) & $-$0.106 & $-$0.122 & X\\
|AF134159|*AA761323|Hs.15106|chromosome 14 open reading frame 1 & 0.105 & 0.081 & 0.096\\
||*AI392710|Hs.355401|Homo sapiens, clone IMAGE:4291796, mRNA, partial cds & 0.085 & 0.078 & X\\
||*AA808306|Hs.252627|ESTs & $-$0.080 & $-$0.085 & $-$0.135\\
|Y10659|*AA398286|Hs.285115|interleukin 13 receptor, alpha 1 & 0.063 & 0.057 & X\\
||||LC\_26716 & 0.057 & $<$ & $<$\\
|D42043|*AA243583|Hs.79123|KIAA0084 protein & $-$0.048 & $<$ & $-$0.207\\
|X59812|*H98765|Hs.82568|cytochrome P450, subfamily XXVIIA (steroid 27-hydroxylase, cerebrotendinous xanthomatosis), polypeptide 1 & $-$0.047 & $-$0.065 & X\\
|AB020662|$\sim$AA054122|Hs.182982|golgin-67 & $<$ & $-$0.073 & X\\
|X66079|*AA490435|Hs.192861|Spi-B transcription factor (Spi-1/PU.1 related) & $<$& 0.053 & X\\
|D88532|*AI219836|Hs.372548|phosphoin ositide-3-kinase, regulatory subunit, polypeptide 3 (p55, gamma) & $<$ & $<$ & 0.110\\
|M29536|*AI052256|Hs.12163|eukaryotic translation initiation factor 2, subunit 2 (beta, 38kD ) & X & X & 0.105\\
|AF127481|*AA262155|Hs.301946|lymphoid blast crisis oncogene & X & X & 0.103\\
|U50196|$\sim$N48691|Hs.94382|adenosine kinase & X & X & $-$0.095\\
|M20430|*AA714513|Hs.352392|major histocompatibility complex, class II, DR beta 5 & X & X & $-$0.091\\
|AF414120|*AA210929|Hs.247824|cytotoxic T-lymphocyte-associated protein 4 & X & X & $-$0.090
\\\hline\hline
\end{tabular}
\end{center}
\end{table}

\spacingset{1.5}
We use the concordance (C)-index \citep{Harrell_etal_1982, Harrell_etal_1996} to evaluate the predictive performance of the procedures. It is the proportion of all usable patient pairs whose predictions and outcomes are concordant. Note that due to censoring, patient pairs may not be usable. C-index values are between $0$ and $1$, and $0.5$ would mean that a fitted model has no predictive discrimination. We randomly split the dataset into training and testing sets in 0.8 and 0.2 proportion. We then screen the top $1500$ predictors using the training set, before applying the three procedures. The out-of-sample C-index is evaluated on the testing set. We repeat this whole procedure $1000$ times, and the average values along with standard errors are given in Table~\ref{tab_C-index}. We observe that FarmHazard yields a larger out-of-sample C-index than LASSO, hence has better predictive performance.

\spacingset{1.8}
\begin{table}[H]
\begin{center}
\caption{\small{Out-of-sample C-index.}}
\label{tab_C-index}
\footnotesize
\vspace{2mm}
\begin{tabular}{cccc}
\hline\hline
 & {FarmHazard-L} & {FarmHazard-S} & {LASSO}
\\\cline{1-4}
{Average} & 0.624 & 0.623 & 0.545\\
{Standard error} & 0.002 & 0.002 & 0.002
\\\hline\hline
\end{tabular}
\end{center}
\end{table}

\spacingset{1.5}

\section{Conclusion}
A stylized feature for high-dimensional data is the dependence of the measurements.
In this paper, we proposed FarmHazard, a new model extending Cox's proportional hazards model that is able to deal with high-dimensional correlated covariates and hence address one of the most important issues in the analysis of big data. We overcame the additional technical challenges of Cox's model emerging from censored data and time-dependent covariates. The new procedures, FarmHazard-L and FarmHazard-S, generated by the new model decompose the high-dimensional covariates via common factors and idiosyncratic components. They learn these factors and components from the data and use them as new predictors, improving upon the usual Cox's model regularization techniques.
Another considerable benefit of using this new set of predictors is the ability to perform screening in the presence of highly correlated covariates for ultra-high dimensional problems.

\newpage

\appendix

\begin{center}
\textbf{\LARGE{Appendix}}
\end{center}

\section{Proof of Lemma~\ref{lem:factor}}\label{sec:proof-lem-factor}

The first four bounds (corresponding to the general and the sub-Gaussian cases) come from a combination of Theorem 10.4 and Corollary 10.2 in \citet{SFDS}. Regarding the improved bounds when $\|\bx^{(2)}\|_{\infty}$ and $\|\bff\|_{\infty}$ are bounded and $\|\bB^{\top} \bu\|_{\infty} \leq \mathcal{C} \sqrt{p_2}$ for some $\mathcal{C}>0$, we write for each $i \in [n]$,
$\|\bx^{(2)}_i\|_2 \leq \sqrt{p_2} \|\bx^{(2)}_i\|_{\infty}$,
$\|\bff_i\|_2 \leq \sqrt{K} \|\bff_i\|_{\infty}$, and
$\|\bB^{\top} \bu_i\|_2 \leq \sqrt{K} \|\bB^{\top} \bu_i\|_{\infty}$.

\section{Proof of Theorem~\ref{thm:consistency-estimated-factors}}\label{sec:proof-thm:consistency-estimated-factors}

We prove Theorem~\ref{thm:consistency-estimated-factors-btheta} below, from which Theorem~\ref{thm:consistency-estimated-factors} naturally follows.

\begin{thm}\label{thm:consistency-estimated-factors-btheta}
	Suppose that Assumptions~\ref{assump:RSC-1}--\ref{assump:factor} hold.
	
Define $M=6 |\cS|^{3/2} \sup\limits_{t\in [0, \tau]} \|\hat{\bW}(t)\|_{\max}^3$. Let $\hat{\btheta}^{\,\star}=
((\bbeta^{\star})^{\top},
(\hat{\bgamma}^{\,\star}_2)^{\top})^{\top} \in \R^{p+K}$, where $\hat{\bgamma}^{\,\star}_2 = \hat{\bB}^{\top} \bbeta_2^{\star}$, and
\begin{align*}
\eta = \left\|\frac{1}{n} \sum_{i=1}^n \int_0^{\tau} \left\{ \hat{\bw}_i(t) - \frac{\sum_{j=1}^n Y_j(t) \hat{\bw}_j(t) \exp(\bx_j(t)^{\top} \bbeta^{\star})}{\sum_{j=1}^n Y_j(t) \exp(\bx_j(t)^{\top} \bbeta^{\star})} \right\} d N_i(t) \right\|_{\infty},	
\end{align*}		
where $\{\hat{\bw}_i(t)^{\top}\}_{i\in [n]}$ are the rows of $\hat{\bW}(t) = (\bX_1(t), \hat{\bW}_2)$.
If $\frac{7}{\mu} \eta < \lambda < \frac{\kappa_2\kappa_{\infty}\mu}{12M\sqrt{|\cS|}}$, then the solution of~\eqref{L1-ppmle} satisfies
$\supp(\hat{\btheta})\subseteq \supp(\hat{\btheta}^{\,\star})$ and
\begin{eqnarray*}
\|\hat{\btheta}-\hat{\btheta}^{\,\star}\|_{\infty}\leq \frac{6\lambda}{5\kappa_{\infty}}, \qquad  \|\hat{\btheta}-\hat{\btheta}^{\,\star}\|_2\leq \frac{4\lambda\sqrt{|\cS|}}{\kappa_2},
\qquad \|\hat{\btheta}-\hat{\btheta}^{\,\star}\|_1\leq \frac{6\lambda|\cS|}{5\kappa_{\infty}}.
\end{eqnarray*}
If there exists $C>7$ such that $\eta < \frac{\kappa_2\kappa_{\infty}\mu^2}{12CM\sqrt{|\cS|}}$ and
$\min\limits_{j \in \cS_{\bbeta^{\star}}}|\beta^{\star}_j| > \frac{6C}{5\kappa_{\infty}\mu}\eta$, then for $\frac{7}{\mu}\eta < \lambda < \frac{C}{\mu}\eta$, we have $\sign(\hat{\bbeta})=\sign(\bbeta^{\star})$.
Moreover, if $\min\limits_{j \in \supp(\hat{\bgamma}^{\,\star}_2)}|(\hat{\gamma}^{\,\star}_2)_j| > \frac{6C}{5\kappa_{\infty}\mu}\eta$, then $\sign(\hat{\btheta})=\sign(\hat{\btheta}^{\,\star})$.
\end{thm}

We have $\|\hat{\bbeta}-\bbeta^{\star}\|= \|\hat{\btheta}_{[p]}-\hat{\btheta}^{\,\star}_{[p]}\|\leq \|\hat{\btheta}-\hat{\btheta}^{\,\star}\|$ for any norm $\|\cdot\|$, and $\supp(\hat{\bbeta})=\supp(\hat{\btheta}_{[p]})$, hence Theorem~\ref{thm:consistency-estimated-factors} will follow from Theorem~\ref{thm:consistency-estimated-factors-btheta}.
As $\hat{\bW}(t) \hat{\btheta}^{\,\star} = \bX(t) \bbeta^{\star}$ for all $t\in [0, \tau]$, we obtain $\eta = \|\nabla \cL(\hat{\bW}\hat{\btheta}^{\,\star})\|_{\infty}$. We then have $\|\nabla_{\cS} \cL(\hat{\bW}\hat{\btheta}^{\,\star})\|_{\infty}\leq\eta$, $\|\nabla_{\cS} \cL(\hat{\bW}\hat{\btheta}^{\,\star})\|_2\leq\eta\sqrt{|\cS|}$ and $\|\nabla_{\cS} \cL(\hat{\bW}\hat{\btheta}^{\,\star})\|_1\leq\eta|\cS|$.

We will prove Lemma~\ref{lem:bridge} below, which will imply that all the regularity conditions in Theorem B.1 of \citet{Fan_etal_2020} (supplement file) are satisfied, from which the results in Theorem~\ref{thm:consistency-estimated-factors-btheta} will follow. Before introducing and proving Lemma~\ref{lem:bridge}, we first introduce another useful lemma, which will be used in the subsequent proofs.

\begin{lem}[\citet{Fan_etal_2020}]\label{inverse-perturbation}
	Suppose $\bA\in\R^{q\times r}$ and $\bB,\bC\in\R^{r\times r}$ and $\|\bC\bB^{-1}\|<1$, where $\|\cdot\|$ is an induced norm. Then $\|\bA[(\bB+\bC)^{-1}-\bB^{-1}]\|\leq \frac{\|\bA\bB^{-1}\|\cdot\|\bC\bB^{-1}\|}{1-\|\bC\bB^{-1}\|}$.
\end{lem}

\begin{lem}\label{lem:bridge}
	Suppose that Assumptions~\ref{assump:RSC-1}--\ref{assump:factor} hold.	
	Define $M=6 |\cS|^{3/2}
	\sup\limits_{t\in [0, \tau]} \|\hat{\bW}(t)\|_{\max}^3$. Then the next four inequalities hold.
	\begin{align}
	(i) \quad &\|\nabla_{\cdot \cS}^2 \cL(\hat{\bW}\btheta)-\nabla_{\cdot \cS}^2 \cL(\hat{\bW}\hat{\btheta}^{\,\star})\|_{\infty}
	\leq M\|\btheta-\hat{\btheta}^{\,\star}\|_2,\;\text{if } \supp(\btheta)\subseteq \cS,\label{lem:bridge-1}\\
	(ii) \quad &\|(\nabla^2_{\cS \cS}\cL(\hat{\bW}\hat{\btheta}^{\,\star}))^{-1}\|_{\infty}\leq \frac{1}{2\kappa_{\infty}}\label{lem:bridge-2},\\
	(iii) \quad &\|(\nabla^2_{\cS \cS}\cL(\hat{\bW}\hat{\btheta}^{\,\star}))^{-1}\|_2\leq \frac{1}{2\kappa_2},\label{lem:bridge-3}\\
	(iv) \quad & \|\nabla^2_{\cS^{c} \cS} \cL(\hat{\bW}\hat{\btheta}^{\,\star})(\nabla^2_{\cS \cS} \cL(\hat{\bW}\hat{\btheta}^{\,\star}))^{-1}\|_{\infty}\leq 1-\mu.\label{lem:bridge-4}
	\end{align}
	
\end{lem}

\begin{proof}
    \underline{Proof of~\eqref{lem:bridge-1}}.
    For simplicity of notation, we define
	\begin{equation*}
	M_{\nabla V} =
	\sup_{t \in [0, \tau]} \sup_{\btheta\in \R^{p+K}} \sup_{(j,k) \in [p+K]\times \cS} \| \nabla_{\btheta} (\bV(\hat{\bW},\btheta,t)_{j,k}) \|_{\infty}.
	\end{equation*}
	For any $(j,k)\in[p+K]\times \cS$
	and $\btheta\in \R^{p+K}$ satisfying $\supp(\btheta)\subseteq \cS$,
	\begin{equation*}
	\begin{split}
	|\nabla^2_{jk}\cL(\hat{\bW}\btheta)-\nabla^2_{jk}\cL(\hat{\bW}\hat{\btheta}^{\,\star})|
	& = \left| \frac{1}{n} \int_0^{\tau} \{ \bV(\hat{\bW},\btheta,t) - \bV(\hat{\bW},\hat{\btheta}^{\,\star},t) \}_{jk} d \overline{N}(t) \right| \\
	& \leq \frac{1}{n} \int_0^{\tau} M_{\nabla V} \| \btheta - \hat{\btheta}^{\,\star} \|_1 d \overline{N}(t)	\\
	& \leq M_{\nabla V} \| \btheta - \hat{\btheta}^{\,\star} \|_1 \\
	& \leq M_{\nabla V} \sqrt{|\cS|} \| \btheta - \hat{\btheta}^{\,\star} \|_2.
	\end{split}
	\end{equation*}
	Hence,
	\begin{equation}
	\label{bound-for-(i)}
	\begin{split}
	\|\nabla^2_{\cdot \cS}\cL(\hat{\bW}\btheta)-\nabla^2_{\cdot \cS}\cL(\hat{\bW}\hat{\btheta}^{\,\star})\|_{\infty}
	& = \max_{j\in[p+K]}
	\|\nabla^2_{j \cS}\cL(\hat{\bW}\btheta)-\nabla^2_{j \cS}\cL(\hat{\bW}\hat{\btheta}^{\,\star})\|_{1}\\
	& \leq |\cS|^{3/2} M_{\nabla V} \|\btheta-\hat{\btheta}^{\,\star}\|_2.
	\end{split}
	\end{equation}
	We now prove an upper bound for $M_{\nabla V}$. First, we write
	\begin{align*}
	\forall \tilde{\btheta} \in \R^{p+K}, \bV(\hat{\bW},\tilde{\btheta},t)_{jk}
	& = \sum_{i=1}^n \alpha_i(\tilde{\btheta},t) \hat{w}_{ij}(t) \hat{w}_{ik}(t) \\
	& \qquad - \left(\sum_{i=1}^n \alpha_i(\tilde{\btheta},t) \hat{w}_{ij}(t)\right) \left(\sum_{i=1}^n \alpha_i(\tilde{\btheta},t) \hat{w}_{ik}(t)\right),
	\end{align*}
	where
	\begin{equation*}
	\alpha_i(\tilde{\btheta},t) = \frac{Y_i(t) \exp(\hat{\bw}_i(t)^{\top} \tilde{\btheta})}{\sum_{\ell=1}^n Y_{\ell}(t) \exp(\hat{\bw}_{\ell}(t)^{\top} \tilde{\btheta})}.
	\end{equation*}
	The gradient of $\alpha_i(\tilde{\btheta},t)$ with respect to $\btheta$ is
	\begin{equation*}
	\nabla_{\btheta} \alpha_i(\tilde{\btheta},t)
	= \hat{\bw}_i(t) \alpha_i(\tilde{\btheta},t) - \alpha_i(\tilde{\btheta},t) \sum_{\ell=1}^n \alpha_{\ell}(\tilde{\btheta},t) \hat{\bw}_{\ell}(t).
	\end{equation*}
	Consequently, $\nabla_{\btheta} (\bV(\hat{\bW},\tilde{\btheta},t)_{jk})$ equals the following expression
	\begin{equation}\label{eq:grad-V}
	\begin{split}
	& \sum_{i=1}^n \left(\hat{\bw}_i(t) \alpha_i(\tilde{\btheta},t) - \alpha_i(\tilde{\btheta},t) \sum_{\ell=1}^n \alpha_{\ell}(\tilde{\btheta},t) \hat{\bw}_{\ell}(t) \right) \hat{w}_{ij}(t) \hat{w}_{ik}(t) \\
	& \qquad - \left(\sum_{i=1}^n \left(\hat{\bw}_i(t) \alpha_i(\tilde{\btheta},t) - \alpha_i(\tilde{\btheta},t) \sum_{\ell=1}^n \alpha_{\ell}(\tilde{\btheta},t) \hat{\bw}_{\ell}(t)\right) \hat{w}_{ij}(t) \right) \left(\sum_{i=1}^n
	\alpha_i(\tilde{\btheta},t) \hat{w}_{ik}(t) \right) \\
	& \qquad - \left(\sum_{i=1}^n \alpha_i(\tilde{\btheta},t) \hat{w}_{ij}(t) \right) \left(\sum_{i=1}^n \left(\hat{\bw}_i(t) \alpha_i(\tilde{\btheta},t) - \alpha_i(\tilde{\btheta},t) \sum_{\ell=1}^n \alpha_{\ell}(\tilde{\btheta},t) \hat{\bw}_{\ell}(t)\right) \hat{w}_{ik}(t) \right).
	\end{split}
	\end{equation}
	This expression has three terms. The first one can be rewritten as
	\begin{equation*}
	\sum_{i=1}^n \alpha_i(\tilde{\btheta},t) \hat{\bw}_i(t) \hat{w}_{ij}(t) \hat{w}_{ik}(t)
	- \left(\sum_{i=1}^n \alpha_i(\tilde{\btheta},t) \hat{w}_{ij}(t) \hat{w}_{ik}(t)\right)\left(\sum_{\ell=1}^n \alpha_{\ell}(\tilde{\btheta},t) \hat{\bw}_{\ell}(t) \right).
	\end{equation*}
	Using the triangle inequality, we get that the supremum norm of this term is upper bounded by $2 \|\hat{\bW}(t)\|_{\max}^3$, since $\alpha_i(\tilde{\btheta},t)\geq 0$ and $\sum_{i=1}^n \alpha_i(\tilde{\btheta},t) = 1$. We study similarly the next two terms in~\eqref{eq:grad-V} to obtain $\|\nabla_{\btheta} (\bV(\hat{\bW},\tilde{\btheta},t)_{jk})\|_{\infty} \leq 6 \|\hat{\bW}(t)\|_{\max}^3$. Therefore, $M_{\nabla V}\leq 6 \sup\limits_{t\in [0, \tau]} \|\hat{\bW}(t)\|_{\max}^3$.
	Using this result together with~\eqref{bound-for-(i)}, we get~\eqref{lem:bridge-1}.
	
	\bigskip
	
	\underline{Proof of~\eqref{lem:bridge-2}}. For any $k\in[p+K]$,
	\begin{equation}\label{diff-hessian-from-V}
	\|\nabla^2_{k \cS}\cL(\hat{\bW}\hat{\btheta}^{\,\star})-\nabla^2_{k \cS}\cL(\bW\btheta^{\star})\|_{\infty}
	= \left\| \frac{1}{n} \int_0^{\tau} \{ \bV(\bW,\btheta^{\star},t) - \bV(\hat{\bW},\hat{\btheta}^{\,\star},t) \}_{k \cS} d \overline{N}(t) \right\|_{\infty}.
	\end{equation}
	As $\hat{\bW}(t) \hat{\btheta}^{\,\star} = \bX(t) \bbeta^{\star} = \bW(t) \btheta^{\star}$ for all $t\in [0, \tau]$, we can write explicitly
	\begin{equation*}
	\begin{split}
	\{ \bV(\bW,\btheta^{\star},t) - \bV(\hat{\bW},\hat{\btheta}^{\,\star},t) \}_{k \cS}
	& = \frac{\ba_1}{\sum_{\ell=1}^n Y_{\ell}(t) \exp(\bx_{\ell}(t)^{\top} \bbeta^{\star})} \\
	& \qquad - \frac{\ba_2}{\left(\sum_{\ell=1}^n Y_{\ell}(t) \exp(\bx_{\ell}(t)^{\top} \bbeta^{\star})\right)^2},
	\end{split}
	\end{equation*}		
	where
	\begin{equation*}
	\ba_1 = \sum_{i=1}^n Y_i(t) \exp(\bx_i(t)^{\top} \bbeta^{\star}) \{ w_{ik}(t) \bw_{i \cS}(t)^{\top} - \hat{w}_{ik}(t) \hat{\bw}_{i \cS}(t)^{\top} \},
	\end{equation*}
	and
	\begin{equation*}
	\begin{split}
	\ba_2
	& = \left(\sum_{i=1}^n Y_i(t) w_{ik}(t) \exp(\bx_i(t)^{\top} \bbeta^{\star})\right) \left(\sum_{i=1}^n Y_i(t) \bw_{i \cS}(t) \exp(\bx_i(t)^{\top} \bbeta^{\star})\right)^{\top} \\
	& \qquad - \left(\sum_{i=1}^n Y_i(t) \hat{w}_{ik}(t) \exp(\bx_i(t)^{\top} \bbeta^{\star})\right) \left(\sum_{i=1}^n Y_i(t) \hat{\bw}_{i \cS}(t) \exp(\bx_i(t)^{\top} \bbeta^{\star})\right)^{\top} \\
	& = \sum_{i=1}^n \sum_{j=1}^n Y_i(t) Y_j(t) \exp(\bx_i(t)^{\top} \bbeta^{\star}) \exp(\bx_j(t)^{\top} \bbeta^{\star}) \{ w_{ik}(t) \bw_{j \cS}(t)^{\top} - \hat{w}_{ik}(t) \hat{\bw}_{j \cS}(t)^{\top} \}.
	\end{split}
	\end{equation*}
	Define
	\begin{align}
	\alpha_i(\bbeta^{\star},t) & = \frac{Y_i(t) \exp(\bx_i(t)^{\top} \bbeta^{\star})}{\sum_{\ell=1}^n Y_{\ell}(t) \exp(\bx_{\ell}(t)^{\top} \bbeta^{\star})}, \label{eq:alpha_i_appendix} \\
	\alpha_{i,j}(\bbeta^{\star},t) & = \frac{Y_i(t) \exp(\bx_i(t)^{\top} \bbeta^{\star}) Y_j(t) \exp(\bx_j(t)^{\top} \bbeta^{\star})}{\left(\sum_{\ell=1}^n Y_{\ell}(t) \exp(\bx_{\ell}(t)^{\top} \bbeta^{\star})\right)^2} \nonumber.
	\end{align}
	Note that $\alpha_i(\bbeta^{\star},t) \in [0,1]$ and $\sum_{i=1}^n \alpha_i(\bbeta^{\star},t) = 1$. Similarly, $\alpha_{i,j}(\bbeta^{\star},t) \in [0,1]$ and $\sum_{i=1}^n \sum_{j=1}^n \alpha_{i,j}(\bbeta^{\star},t) = 1$.
With this notation, we can write
	\begin{equation*}
	\begin{split}
	\{ \bV(\bW,\btheta^{\star},t) - \bV(\hat{\bW},\hat{\btheta}^{\,\star},t) \}_{k \cS}
	& = \sum_{i=1}^n \alpha_i(\bbeta^{\star},t) \{ w_{ik}(t) \bw_{i \cS}(t)^{\top} - \hat{w}_{ik}(t) \hat{\bw}_{i \cS}(t)^{\top} \} \\
	& \qquad - \sum_{i=1}^n \sum_{j=1}^n \alpha_{i,j}(\bbeta^{\star},t) \{ w_{ik}(t) \bw_{j \cS}(t)^{\top} - \hat{w}_{ik}(t) \hat{\bw}_{j \cS}(t)^{\top} \}.
	\end{split}
	\end{equation*}		
	Therefore,
	\begin{equation*}
	\begin{split}
	\| \{ \bV(\bW,\btheta^{\star},t) - \bV(\hat{\bW},\hat{\btheta}^{\,\star},t) \}_{k \cS} \|_{\infty}
	& \leq \max_{i\in [n]} \| w_{ik}(t) \bw_{i \cS}(t)^{\top} - \hat{w}_{ik}(t) \hat{\bw}_{i \cS}(t)^{\top} \|_{\infty} \\
	& \qquad + \max_{i,j\in [n]} \| w_{ik}(t) \bw_{j \cS}(t)^{\top} - \hat{w}_{ik}(t) \hat{\bw}_{j \cS}(t)^{\top} \|_{\infty}.
	\end{split}
	\end{equation*}
Let $M_W = \sup\limits_{t\in [0, \tau]} \|\bW(t)\|_{\max} = \sup\limits_{t\in [0, \tau]} \|\bX_1(t)\|_{\max} \lor \|\bW_2\|_{\max}$. Then $\|\hat{\bW}(t)\|_{\max}\leq M_W + \varepsilon$.

\vspace{1mm}
	\noindent On the one hand, for any $i \in [n]$, $\|w_{ik}(t)\bw_{i \cS}(t)^{\top} - \hat{w}_{ik}(t)\hat{\bw}_{i \cS}(t)^{\top}\|_{\infty}$ is upper bounded by
	\begin{equation*}
	\begin{split}
	&| w_{ik}(t)|\cdot\|(\hat{\bw}_{i \cS}(t)-\bw_{i \cS}(t))^{\top}\|_{\infty}+|\hat{w}_{ik}(t)-w_{ik}(t)|\cdot\|\hat{\bw}_{i \cS}(t)^{\top}\|_{\infty}\\
	&\leq \|\bW(t)\|_{\max}\cdot\|(\hat{\bw}_{i \cS}(t)-\bw_{i \cS}(t))^{\top}\|_{\infty}+|\hat{w}_{ik}(t)-w_{ik}(t)|\cdot \|\hat{\bW}(t)\|_{\max}\\
	&\leq M_W\|(\hat{\bw}_{i \cS}(t)-\bw_{i \cS}(t))^{\top}\|_{\infty}+(M_W+\varepsilon)|\hat{w}_{ik}(t)-w_{ik}(t)|.
	\end{split}
	\end{equation*}
	On the other hand, for any $i,j \in [n]$, $\|w_{ik}(t)\bw_{j \cS}(t)^{\top} - \hat{w}_{ik}(t)\hat{\bw}_{j \cS}(t)^{\top}\|_{\infty}$ is upper bounded by
	\begin{equation*}
	\begin{split}
	&| w_{ik}(t)|\cdot\|(\hat{\bw}_{j \cS}(t)-\bw_{j \cS}(t))^{\top}\|_{\infty}+|\hat{w}_{ik}(t)-w_{ik}(t)|\cdot\|\hat{\bw}_{j \cS}(t)^{\top}\|_{\infty}\\
	&\leq \|\bW(t)\|_{\max}\cdot\|(\hat{\bw}_{j \cS}(t)-\bw_{j \cS}(t))^{\top}\|_{\infty}+|\hat{w}_{ik}(t)-w_{ik}(t)|\cdot \|\hat{\bW}(t)\|_{\max}\\
	&\leq M_W\|(\hat{\bw}_{j \cS}(t)-\bw_{j \cS}(t))^{\top}\|_{\infty}+(M_W+\varepsilon)|\hat{w}_{ik}(t)-w_{ik}(t)|.
	\end{split}
	\end{equation*}
Consequently,
	\begin{equation*}
	\begin{split}
	\| \{ \bV(\bW,\btheta^{\star},t) - \bV(\hat{\bW},\hat{\btheta}^{\,\star},t) \}_{k \cS} \|_{\infty}
	& \leq 2(2M_W+\varepsilon) \|\hat{\bW}(t) - \bW(t)\|_{\max} \\
	& \leq 2 \varepsilon (2M_W+\varepsilon).
	\end{split}
	\end{equation*}
    Plugging the inequality derived above into~\eqref{diff-hessian-from-V}, we get
	\begin{equation*}
	\begin{split}
	\|\nabla^2_{k \cS}\cL(\hat{\bW}\hat{\btheta}^{\,\star})-\nabla^2_{k \cS}\cL(\bW\btheta^{\star})\|_{\infty}
	& \leq \frac{1}{n} \int_0^{\tau} \| \{ \bV(\bW,\btheta^{\star},t) - \bV(\hat{\bW},\hat{\btheta}^{\,\star},t) \}_{k \cS} \|_{\infty} d \overline{N}(t) \\
	& \leq \frac{1}{n} \int_0^{\tau} 2 \varepsilon (2M_W+\varepsilon) d \overline{N}(t) \\
	& \leq 2 \varepsilon (2M_W+\varepsilon).
	\end{split}
	\end{equation*}
	Therefore,
	\begin{equation}\label{lem:B-inf}
	\begin{split}
	\|\nabla^2_{\cdot \cS}\cL(\hat{\bW}\hat{\btheta}^{\,\star})-\nabla^2_{\cdot \cS}\cL(\bW\btheta^{\star})\|_{\infty}
	&=\max_{k\in[p+K]}
	\|\nabla^2_{k \cS}\cL(\hat{\bW}\hat{\btheta}^{\,\star})-\nabla^2_{k \cS}\cL(\bW\btheta^{\star})\|_{1}\\
	&\leq 2 \varepsilon (2M_W+\varepsilon) |\cS|.
	\end{split}
	\end{equation}
	Let $\alpha=\|(\nabla^2_{\cS \cS}\cL(\bW\btheta^{\star}))^{-1}[\nabla^2_{\cS \cS}\cL(\hat{\bW}\hat{\btheta}^{\,\star})-\nabla^2_{\cS \cS}\cL(\bW\btheta^{\star})]\|_{\infty}$. Then
	\begin{equation}\label{lemma_B_3_alpha}
	\begin{split}
	\alpha
	&\leq \|(\nabla^2_{\cS \cS}\cL(\bW\btheta^{\star}))^{-1}\|_{\infty}\|\nabla^2_{\cS \cS}\cL(\hat{\bW}\hat{\btheta}^{\,\star})-\nabla^2_{\cS \cS}\cL(\bW\btheta^{\star})\|_{\infty}\\
	&\leq \frac{1}{4\kappa_{\infty}}2 \varepsilon (2M_W+\varepsilon) |\cS| \\
	&\leq\frac{\mu}{4},
	\end{split}
	\end{equation}
	the last inequality coming from Assumption~\ref{assump:factor}.
	With Lemma~\ref{inverse-perturbation}, we then obtain
	\begin{equation*}
	\begin{split}
	\|(\nabla^2_{\cS \cS}\cL(\hat{\bW}\hat{\btheta}^{\,\star}))^{-1}-(\nabla^2_{\cS \cS}\cL(\bW\btheta^{\star}))^{-1}\|_{\infty}
	&\leq \|(\nabla^2_{\cS \cS}\cL(\bW\btheta^{\star}))^{-1}\|_{\infty}\frac{\alpha}{1-\alpha}\\
    &\leq\frac{1}{4\kappa_{\infty}},
	\end{split}
	\end{equation*}
	as $\alpha \leq \mu/4 \leq 1/8 \leq 1/2$.
	Combined with Assumption~\ref{assump:RSC-1} and the triangle inequality, the inequality derived above gives
	\begin{equation}\label{lemma_B_3_tii}
	\begin{split}
	&\|(\nabla^2_{\cS \cS}\cL(\hat{\bW}\hat{\btheta}^{\,\star}))^{-1}\|_{\infty}\leq\|(\nabla^2_{\cS \cS}\cL(\bW\btheta^{\star}))^{-1}\|_{\infty}+\frac{1}{4\kappa_{\infty}}\leq \frac{1}{2\kappa_{\infty}}.
	\end{split}
	\end{equation}
	
	\bigskip
	
	\underline{Proof of~\eqref{lem:bridge-3}}. Using~\eqref{lemma_B_3_tii} and the fact that $\|\bA\|_2\leq\|\bA\|_{\infty}$ for any symmetric matrix $\bA$, we obtain
	\begin{equation*}
	\|(\nabla^2_{\cS \cS}\cL(\hat{\bW}\hat{\btheta}^{\,\star}))^{-1}\|_2\leq\frac{1}{2\kappa_{\infty}}\leq\frac{1}{2\kappa_2}.
	\end{equation*}
	
	\bigskip
	
	\underline{Proof of~\eqref{lem:bridge-4}}. We first write
	\begin{equation}
	\begin{split}
	&\|\nabla^2_{\cS^{c} \cS} \cL(\hat{\bW}\hat{\btheta}^{\,\star})(\nabla^2_{\cS \cS} \cL(\hat{\bW}\hat{\btheta}^{\,\star}))^{-1}
	-\nabla^2_{\cS^{c} \cS} \cL(\bW\btheta^{\star})(\nabla^2_{\cS \cS} \cL(\bW\btheta^{\star}))^{-1}\|_{\infty}\\
	&\leq \|\nabla^2_{\cS^{c} \cS} \cL(\hat{\bW}\hat{\btheta}^{\,\star})-\nabla^2_{\cS^{c} \cS} \cL(\bW\btheta^{\star})\|_{\infty}
	\|(\nabla^2_{\cS \cS} \cL(\hat{\bW}\hat{\btheta}^{\,\star}))^{-1}\|_{\infty}\\
	&\qquad +\|\nabla^2_{\cS^{c} \cS} \cL(\bW\btheta^{\star})[(\nabla^2_{\cS \cS} \cL(\hat{\bW}\hat{\btheta}^{\,\star}))^{-1}-(\nabla^2_{\cS \cS} \cL(\bW\btheta^{\star}))^{-1}]\|_{\infty}.
	\label{lem:bridge-4-proof-start}
	\end{split}
	\end{equation}
The first term of the right-hand side is
	\begin{equation}
	\begin{split}
	&\|\nabla^2_{\cS^{c} \cS} \cL(\hat{\bW}\hat{\btheta}^{\,\star})-\nabla^2_{\cS^{c} \cS} \cL(\bW\btheta^{\star})\|_{\infty}
	\|(\nabla^2_{\cS \cS} \cL(\hat{\bW}\hat{\btheta}^{\,\star}))^{-1}\|_{\infty}\leq\frac{1}{\kappa_{\infty}}\varepsilon (2M_W + \varepsilon) |\cS|,
	\label{4-RHS-1}
	\end{split}
	\end{equation}
	by~\eqref{lem:bridge-2} and~\eqref{lem:B-inf}.
	As for the second term, take $\bA=\nabla^2_{\cS^{c} \cS} \cL(\bW\btheta^{\star})$, $\bB=\nabla^2_{\cS \cS} \cL(\bW\btheta^{\star})$ and $\bC=\nabla^2_{\cS \cS} \cL(\hat{\bW}\hat{\btheta}^{\,\star})-\nabla^2_{\cS \cS} \cL(\bW\btheta^{\star})$. By Assumption~\ref{assump:IC}, $\|\bA\bB^{-1}\|_{\infty}\leq 1-2\mu\leq 1$. And~\eqref{lemma_B_3_alpha} gives $\|\bC\|_{\infty}\|\bB^{-1}\|_{\infty}\leq\frac{1}{4\kappa_{\infty}}2 \varepsilon (2M_W+\varepsilon) |\cS|\leq\frac{1}{2}$. Then, with Lemma~\ref{inverse-perturbation},
	\begin{equation}
	\begin{split}
	&\|\nabla^2_{\cS^{c} \cS} \cL(\bW\btheta^{\star})[(\nabla^2_{\cS \cS} \cL(\hat{\bW}\hat{\btheta}^{\,\star}))^{-1}-(\nabla^2_{\cS \cS} \cL(\bW\btheta^{\star}))^{-1}]\|_{\infty}\\
	&=\|\bA[(\bB+\bC)^{-1}-\bB^{-1}]\|_{\infty}\\
	&\leq \|\bA\bB^{-1}\|_{\infty} \frac{\|\bC\bB^{-1}\|_{\infty}}{1-\|\bC\bB^{-1}\|_{\infty}}\\
	&\leq \frac{\|\bC\|_{\infty}\|\bB^{-1}\|_{\infty}}{1-\|\bC\|_{\infty}\|\bB^{-1}\|_{\infty}}\\
	&\leq 2\|\bC\|_{\infty}\|\bB^{-1}\|_{\infty}\\
	&\leq \frac{1}{\kappa_{\infty}} \varepsilon (2M_W+\varepsilon) |\cS|.
	\label{4-RHS-2}
	\end{split}
	\end{equation}
	By~\eqref{lem:bridge-4-proof-start},~\eqref{4-RHS-1} and~\eqref{4-RHS-2}, we get
	\begin{equation*}
	\begin{split}
	&\|\nabla^2_{\cS^{c} \cS} \cL(\hat{\bW}\hat{\btheta}^{\,\star})(\nabla^2_{\cS \cS} \cL(\hat{\bW}\hat{\btheta}^{\,\star}))^{-1}-\nabla^2_{\cS^{c} \cS} \cL(\bW\btheta^{\star})(\nabla^2_{\cS \cS} \cL(\bW\btheta^{\star}))^{-1}\|_{\infty}\\
	&\leq  \frac{2}{\kappa_{\infty}}\varepsilon (2M_W + \varepsilon) |\cS|\\
	&\leq \mu,
	\end{split}
	\end{equation*}
	the second inequality coming from Assumption~\ref{assump:factor}.
	Using this inequality together with Assumption~\ref{assump:IC}, we then obtain
	\begin{equation*}
	\|\nabla^2_{\cS^{c} \cS} \cL(\hat{\bW}\hat{\btheta}^{\,\star})(\nabla^2_{\cS \cS} \cL(\hat{\bW}\hat{\btheta}^{\,\star}))^{-1}\|_{\infty}\leq (1-2\mu)+\mu=1-\mu.
	\end{equation*}
\end{proof}

\section{Proof of Lemma~\ref{lem:gradient-upper-bound}}\label{sec:proof-lem-gradient-upper-bound}

With $\hat{\btheta}^{\,\star}$ defined in Section~\ref{sec:proof-thm:consistency-estimated-factors}, we have
\begin{equation*}
\eta = \displaystyle \|\nabla \cL(\hat{\bW} \hat{\btheta}^{\,\star})\|_{\infty}
= \left\|\frac{1}{n} \sum_{i=1}^n \int_0^{\tau} \left\{ \hat{\bw}_i(t) -
\frac{S^{(1)}(\hat{\bW},\hat{\btheta}^{\,\star},t)}{S^{(0)}(\hat{\bW},\hat{\btheta}^{\,\star},t)} \right\} d N_i(t)\right\|_{\infty},
\end{equation*}
and we can write
\begin{equation*}
\nabla \cL(\hat{\bW} \hat{\btheta}^{\,\star}) = \bxi_1 + \bxi_2 + \bxi_3,
\end{equation*}
where
\begin{align*}
\bxi_1
& = -\frac{1}{n} \sum_{i=1}^n \int_0^{\tau} \left\{ \bw_i(t) - \frac{S^{(1)}(\bW,\btheta^{\star},t)}{S^{(0)}(\bW,\btheta^{\star},t)} \right\} d N_i(t),\cr
\bxi_2
& = -\frac{1}{n} \sum_{i=1}^n \int_0^{\tau} \big\{ \hat{\bw}_i(t) - \bw_i(t) \big\} d N_i(t),\cr
\bxi_3
& = -\frac{1}{n} \sum_{i=1}^n \int_0^{\tau} \left\{ \frac{S^{(1)}(\bW,\btheta^{\star},t)}{S^{(0)}(\bW,\btheta^{\star},t)} - \frac{S^{(1)}(\hat{\bW},\hat{\btheta}^{\,\star},t)}{S^{(0)}(\hat{\bW},\hat{\btheta}^{\,\star},t)} \right\} d N_i(t).
\end{align*}
Note that $\bxi_1=\nabla \cL(\bW \btheta^{\star})$. Lemma 3.3 in \citet{Huang_etal_2013} gives, for any $x>0$,
\begin{equation*}
\mathbb{P}(\|\nabla \cL(\bW \btheta^{\star})\|_{\infty} > 2 \sup\limits_{t\in [0, \tau]} \|\bW(t)\|_{\max}\, x) \leq 2(p+K) e^{-n x^2 / 2}.
\end{equation*}
Hence
\begin{align*}
\|\nabla \cL(\bW \btheta^{\star})\|_{\infty}
& = O_{\mathbb{P}}\left(\sqrt{\frac{\log (p+K)}{n}}\sup\limits_{t\in [0, \tau]} \|\bW(t)\|_{\max}\right) \\
& = O_{\mathbb{P}}\left(
\sqrt{\frac{\log (p+K)}{n}}\sup\limits_{t\in [0, \tau]} \|\bX_1(t)\|_{\max}
\lor
\|\bW_2\|_{\max}
\right).
\end{align*}
As $\hat{\bW}(t) \hat{\btheta}^{\,\star} = \bX(t) \bbeta^{\star} = \bW(t) \btheta^{\star}$ for all $t\in [0, \tau]$, we have
\begin{equation*}
\frac{S^{(1)}(\bW,\btheta^{\star},t)}{S^{(0)}(\bW,\btheta^{\star},t)} - \frac{S^{(1)}(\hat{\bW},\hat{\btheta}^{\,\star},t)}{S^{(0)}(\hat{\bW},\hat{\btheta}^{\,\star},t)}
= \sum_{i=1}^n \alpha_i(\bbeta^{\star},t) \big\{ \bw_i(t) - \hat{\bw}_i(t) \big\},
\end{equation*}
where $\alpha_i(\bbeta^{\star},t)$ is defined in~\eqref{eq:alpha_i_appendix}.
Therefore, by triangle inequality, we obtain that $\|\bxi_2\|_{\infty}$ and $\|\bxi_3\|_{\infty}$ are upper bounded by $\|\hat{\bW}_2 - \bW_2\|_{\max}$.
This is controlled with Lemma~\ref{lem:factor}, from which Remark~\ref{rem-gradient-upper-bound} follows.

\section{Proof of Lemma~\ref{lem:screening-population}}\label{sec:proof-lem-screening-population}

Recall that
\begin{equation}\label{population-par}
\bd_j(\beta,\bgamma)
= -\int_0^{\tau} \left\{ r_j^{(1)}(t) - \frac{s_j^{(1)}(\beta,\bgamma,t)}{s_j^{(0)}(\beta,\bgamma,t)} r_j^{(0)}(t)\right\} dt,
\end{equation}
and $(\beta_j, \bgamma_j)$ is the solution of $\bd_j(\beta_j,\bgamma_j) = \mathbf{0}_{1+K}$.

Let $S_T(\cdot\mid \bx)$ and $S_C(\cdot\mid \bx)$ be the conditional survival functions of the survival time $T$ and the censoring time $C$, respectively, and $F_T(\cdot\mid \bx)$ be the conditional cumulative distribution function of $T$, given the covariate vector $\bx$. The following lemma gives the first dimension of $\bd_j$ in terms of these functions.

\begin{lem}\label{lem:population-par}
The first dimension $d_{j1}(\beta, \bgamma)$ of $\bd_j(\beta,\bgamma)$ in ~\eqref{population-par} is given by
\begin{equation*}
\begin{split}
d_{j1}(\beta,\bgamma)
& = - \cov(u_{1j}, \E [ F_T(C\mid \bx) \mid \bx ]) \\
& \qquad + \int_0^{\tau} \frac{\E [ u_{1j} \exp(u_{1j} \beta + \bff_1^{\top} \bgamma) S_T(t\mid \bx) S_C(t\mid \bx) ]}{\E [ \exp(u_{1j} \beta + \bff_1^{\top} \bgamma) S_T(t\mid \bx) S_C(t\mid \bx) ]} \E [ \lambda(t\mid \bx) S_T(t\mid \bx) S_C(t\mid \bx) ] dt.
\end{split}
\end{equation*}
\end{lem}

\medskip

\begin{proof}
First, the function $d_{j1}(\beta, \bgamma)$ is defined for any $(\beta, \bgamma) \in \R^{1+K}$ as
\begin{equation*}
d_{j1}(\beta, \bgamma) = \int_0^{\tau} \left\{ - \E [ Y(t) u_{1j} \lambda(t\mid \bx) ] + \frac{\E [ Y(t) u_{1j} \exp(u_{1j} \beta + \bff_1^{\top} \bgamma) ]}{\E [ Y(t) \exp(u_{1j} \beta + \bff_1^{\top} \bgamma) ]} \E [ Y(t) \lambda(t\mid \bx) ] \right\} dt.
\end{equation*}
By definition of $Y(t)$ and independence of $C$ and $T$ conditional on $\bx$, we get
\begin{align*}
& \E [ Y(t)\mid \bx ]
= S_T(t\mid \bx) S_C(t\mid \bx),\\
& \E [ Y(t) \lambda(t\mid \bx) ]
= \E [ \lambda(t\mid \bx) S_T(t\mid \bx) S_C(t\mid \bx) ],\\
& \E [ Y(t) u_{1j} \lambda(t\mid \bx) ]
= \E [u_{1j} \lambda(t\mid \bx) S_T(t\mid \bx) S_C(t\mid \bx) ],\\
& \E [ Y(t) u_{1j} \exp(u_{1j} \beta + \bff_1^{\top} \bgamma) ]
= \E [u_{1j} \exp(u_{1j} \beta + \bff_1^{\top} \bgamma) S_T(t\mid \bx) S_C(t\mid \bx) ],\\
& \int_0^{\tau} \E [ Y(t) u_{1j} \lambda(t\mid \bx) ] dt
= \int_0^{\tau} \E [ u_{1j} \lambda(t\mid \bx) S_T(t\mid \bx) S_C(t\mid \bx) ] dt
= \cov(u_{1j}, \E [ F_T(C\mid \bx)\mid \bx ]).
\end{align*}
Hence we obtain the desired result.
\end{proof}

\medskip

As $d_{j1}(\beta_j, \bgamma_j) = 0$, we obtain
\begin{equation}
\label{eq:lem-screening-mvt}
| d_{j1}(0,\bgamma_j) | = | d_{j1}(\beta_j, \bgamma_j) - d_{j1}(0,\bgamma_j) |
= \left| \frac{\partial d_{j1}}{\partial \beta}(\check{\beta}_j, \bgamma_j) \cdot \beta_j \right|,
\end{equation}
for some $\check{\beta}_j$ between zero and $\beta_j$, by the mean value theorem. On the one hand, we can use the following majorations
\begin{equation*}
\begin{split}
\left| \frac{\partial d_{j1}}{\partial \beta}(\check{\beta}_j, \bgamma_j) \right|
& = \left| \int_0^{\tau} \left\{ \frac{\E [ Y(t) u_{1j}^2 e^{u_{1j} \check{\beta}_j + \bff_1^{\top} \bgamma_j} ]}{\E [ Y(t) e^{u_{1j} \check{\beta}_j + \bff_1^{\top} \bgamma_j} ]} - \left(\frac{\E [ Y(t) u_{1j} e^{u_{1j} \check{\beta}_j + \bff_1^{\top} \bgamma_j} ]}{\E [ Y(t) e^{u_{1j} \check{\beta}_j + \bff_1^{\top} \bgamma_j} ]} \right)^2 \right\} \E [ Y(t) \lambda(t\mid \bx) ] dt \right| \\
& \leq 2 M_0^2 \int_0^{\tau} \E [ Y(t) \lambda(t\mid \bx) ] dt \\
& = 2 M_0^2 \E [ \E [ S_C(T\mid \bx) \mid \bx ] ].
\end{split}
\end{equation*}
As $S_C \leq 1$, this gives
\begin{equation}
\label{eq:lem-screening-upper}
\left| \frac{\partial d_{j1}}{\partial \beta}(\check{\beta}_j, \bgamma_j) \right|
\leq 2 M_0^2.
\end{equation}
On the other hand, we have
\begin{equation*}
\begin{split}
d_{j1}(0, \bgamma_j)
& = - \cov(u_{1j}, \E [ F_T(C\mid \bx) \mid \bx ]) \\
& \qquad + \int_0^{\tau} \frac{\E [ u_{1j} \exp(\bff_1^{\top} \bgamma_j) S_T(t\mid \bx) S_C(t\mid \bx) ]}{\E [\exp(\bff_1^{\top} \bgamma_j) S_T(t\mid \bx) S_C(t\mid \bx) ]} \E [ \lambda(t\mid \bx) S_T(t\mid \bx) S_C(t\mid \bx) ] d t.
\end{split}
\end{equation*}
$S_T(t\mid \bx) S_C(t\mid \bx)$ is the probability of being at risk at time $t$, and $\cov(u_{1j}, \exp(\bff_1^{\top} \bgamma_j) S_T(t\mid \bx) S_C(t\mid \bx))$ and $\cov(u_{1j}, \E [ F_T(C\mid \bx) \mid \bx ])$ have opposite signs for each $j \in \supp(\bbeta^{\star})$ and $t \in [0, \tau]$, as in appendix D of \citet{Zhao_Li_2012}.
Therefore,
\begin{equation}
\label{eq:lem-screening-lower}
\begin{split}
|d_{j1}(0, \bgamma_j)|
&\geq \left|\cov(u_{1j}, \E [ F_T(C\mid \bx) \mid \bx ]) \right|.
\end{split}
\end{equation}
With~\eqref{eq:lem-screening-mvt},~\eqref{eq:lem-screening-upper} and~\eqref{eq:lem-screening-lower}, we therefore get
\begin{equation*}
|\beta_j| \geq \frac{1}{2} M_0^{-2} | \cov (u_{1j}, \E[ F_T(C\mid \bx) \mid \bx ]) |.
\end{equation*}

\section{Proof of Theorem~\ref{thm:screening}}\label{sec:proof-thm-screening}

Using Lemma~\ref{lem:screening-population} and the condition $\xi \leq \nu \min_{j \in \supp(\bbeta^{\star})} \frac{1}{2} M_0^{-2} | \cov (u_{1j}, \E[ F_T(C\mid \bx) \mid \bx ]) |$,
we first obtain
\begin{equation}
\label{eq:proof-screening-1}
\xi \leq \nu \min_{j \in \supp(\bbeta^{\star})} |\beta_j|.
\end{equation}
Under the assumptions in Remark~\ref{rem-gradient-upper-bound} and the condition
\begin{equation*}
\min_{j \in \supp(\bbeta^{\star})}
| \cov (u_{1j}, \E[ F_T(C\mid \bx) \mid \bx ]) | \gg \sqrt{(\log p)/n} + 1/\sqrt{p},
\end{equation*}
we can prove that
\begin{align*}
\max_{j \in \supp(\bbeta^{\star})} \| \bD_j(\beta_j,\bgamma_j) \|_{\infty}
& = O_{\mathbb{P}}\left( \sqrt{(\log p)/n} + 1/\sqrt{p} \right) \\
& = o_{\mathbb{P}} \Big(\min_{j \in \supp(\bbeta^{\star})}
| \cov (u_{1j}, \E[ F_T(C\mid \bx) \mid \bx ]) | \Big),
\end{align*}
where $\bD_j$ is defined in~\eqref{estim-equ}.
Moreover, the same strategy as in the proof of Theorem~\ref{thm:consistency-estimated-factors} gives the existence of a constant $C'>0$ such that
\begin{align*}
| \hat\beta_j - \beta_j |
\leq C' \| \bD_j(\beta_j,\bgamma_j) \|_{\infty},\; \forall j\in [p].
\end{align*}
Further using Lemma~\ref{lem:screening-population}, we then obtain
\begin{equation}
\label{eq:proof-screening-2}
\max_{j \in \supp(\bbeta^{\star})} |\hat\beta_j - \beta_j|
\leq C' \max_{j \in \supp(\bbeta^{\star})}\| \bD_j(\beta_j,\bgamma_j) \|_{\infty}
= o_{\mathbb{P}} \Big(\min_{j \in \supp(\bbeta^{\star})} |\beta_j|\Big).
\end{equation}
Then,~\eqref{eq:proof-screening-1} and~\eqref{eq:proof-screening-2} yield $\mathbb{P}(\supp(\bbeta^{\star}) \subseteq \{ j : | \hat\beta_j | \geq \xi \}) \to 1$.

\newpage

\bibliographystyle{apalike}
\bibliography{../references}

\end{document}